\documentclass{article}
\PassOptionsToPackage{numbers,compress}{natbib}
\usepackage[main,final]{neurips_2025}
\usepackage{booktabs}       
\usepackage{amsfonts}       
\usepackage{nicefrac}       
\usepackage{microtype}      
\usepackage[dvipsnames]{xcolor}         
\usepackage{amsmath}
\usepackage[utf8]{inputenc} 
\usepackage[T1]{fontenc}    
\usepackage[hidelinks]{hyperref}       
\usepackage{url}            
\usepackage{amsfonts}       
\usepackage{graphicx}
\usepackage{algorithm}
\usepackage{algorithmic}
\usepackage{amsthm}
\usepackage{dsfont}
\usepackage{bm} 
\usepackage{colonequals}
\usepackage{enumitem}
\usepackage{wrapfig}

\def \bx { \bm{x } }
\def \bX { \mathbf{X } }

\def \bZ { \mathbf{Z } }

\def \bX { \mathbf{X } }
\def \bY { \mathbf{Y } }
\def \by { \bm{y } }

\DeclareMathOperator*{\argmax}{argmax}

\newtheorem{lemma}{Lemma}
\newtheorem{proposition}{Proposition}

\title{Covariate-moderated Empirical Bayes Matrix Factorization}

\author{William R. P. Denault, Karl Tayeb, Peter Carbonetto \&
Matthew Stephens \\
Departments of Statistics and Human Genetics \\
University of Chicago \\
Chicago, IL 60637, USA \\
\texttt{\{wdenault,ktayeb,pcarbo,mstephens\}@uchicago.edu} \\
\AND
Jason Willwerscheid \\
Mathematics and Computer Science \\
Providence College \\
Providence, RI 02918, USA \\
\texttt{jwillwer@providence.edu} 
}

\begin{document}

\maketitle

\begin{abstract}
Matrix factorization is a fundamental method in statistics and machine
learning for inferring and summarizing structure in multivariate
data. Modern data sets often come with ``side information'' of various
forms (images, text, graphs) that can be leveraged to improve
estimation of the underlying structure. However, existing methods that
leverage side information are limited in the types of data they can
incorporate, and they assume specific parametric models. Here, we
introduce a novel method for this problem, {\em covariate-moderated
  empirical Bayes matrix factorization} (cEBMF). cEBMF is a modular
framework that accepts any type of side information that is
processable by a probabilistic model or a neural network.  The cEBMF
framework can accommodate different assumptions and constraints on the
factors through the use of different priors, and it adapts these
priors to the data. We demonstrate the benefits of cEBMF in
simulations and in analyses of spatial transcriptomics and
collaborative filtering data. A PyTorch-based implementation of cEBMF
with flexible priors is available at
\url{https://github.com/william-denault/cebmf_torch}.
\end{abstract}

\section{Introduction}
\label{submission}
 
Matrix factorization methods, which include principal component
analysis (PCA), factor analysis, and non-negative matrix factorization
(NMF) \citep{lee_learning_1999, lee-seung-2001,
  gillis_nonnegative_2021}, are very widely used methods for inferring
latent structure from data, performing exploratory data analyses, and
visualizing large data sets (e.g., \citep{sainburg_finding_2020,
  alexander_capturing_2022, novembre_interpreting_2008}). Matrix
factorization methods are also instrumental in other statistical
analyses such as adjusting for unobserved confounding
\citep{leek_capturing_2007}.
%
%
When factorizing a matrix, say ${\bf Z}$, the matrix may be
accompanied by additional row or column data---``side
information''---that may be able to ``guide'' the matrix factorization
algorithm toward a more accurate or interpretable factorization. A
recent prominent example of this in genomics research is spatial
transcriptomics data \citep{marx_method_2021}, which is expression
profiled in many genes at many spatial locations (``pixels'')
\citep{vandereyken_methods_2023}. For a variety of reasons, one
typically seeks to factorize ${\bf Z}$, the matrix of gene expression
profiles. But the 2-d coordinates of the pixels also provide important
information about the biological context of the cells; for example, we
might expect nearby pixels to belong to the same cell type or tissue
region. Therefore, ``spatially aware'' matrix factorization methods
have recently been proposed for spatial transcriptomics data
\citep{mefisto, shang_spatially_2022, townes_nonnegative_2023}.


\begin{figure}[t]
\centering
\includegraphics[width=\textwidth]{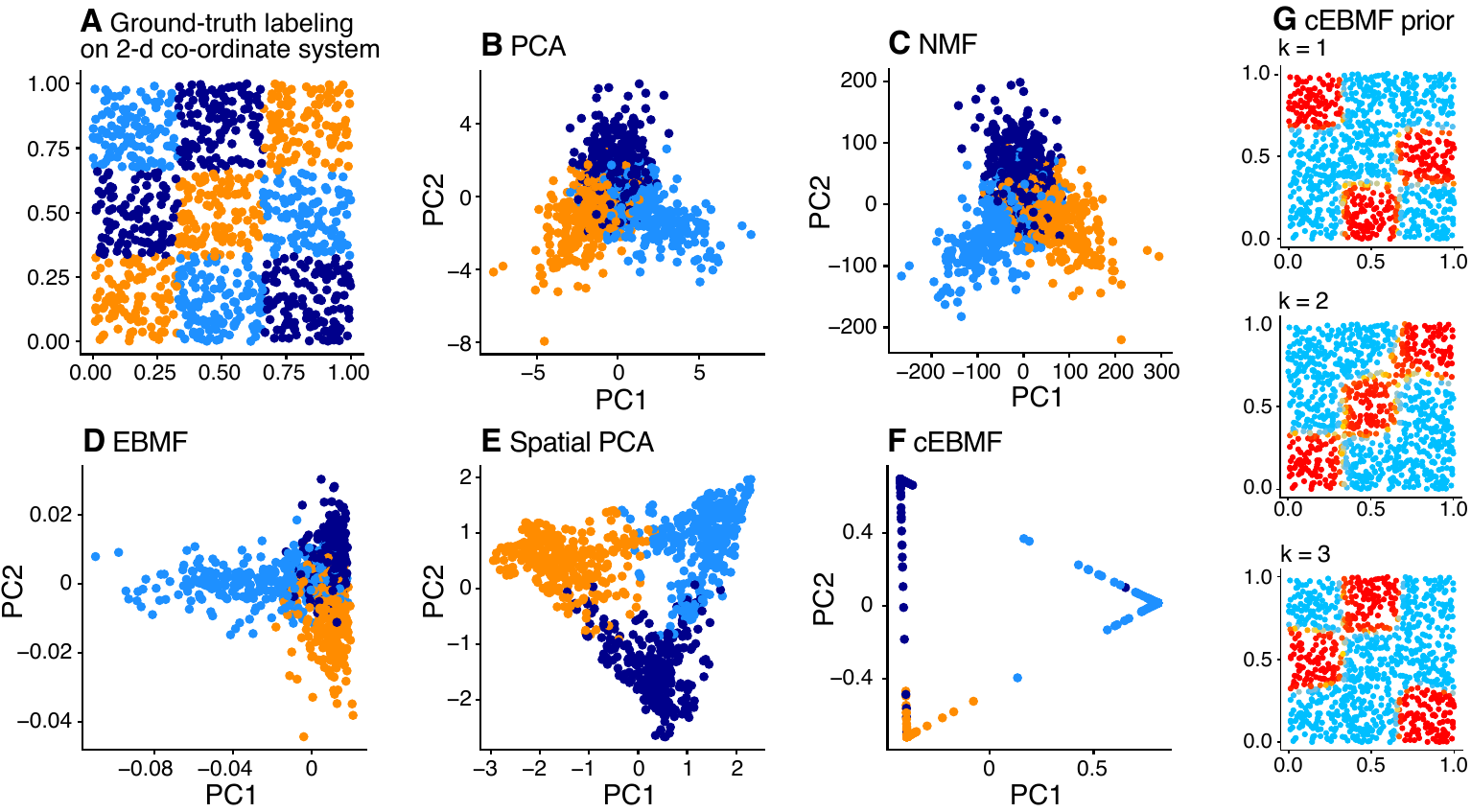}
\label{fig:toy_example}
\vspace*{-1em}
\caption{Toy simulation illustrating cEBMF for learning a matrix
  factorization, ${\bf Z} \approx {\bf L} {\bf F}^T$. In this example,
  ${\bf Z}$ is a $\mbox{1,000} \times 200$ matrix. Each of the $n =
  \mbox{1,000}$ data points is assigned to one of three clusters
  ({\color{Orange} \bf orange}, {\color{SkyBlue} \bf light blue},
  {\color{BlueViolet} \bf dark blue}). Points near each other tend to
  be assigned to the same cluster, except near the boundaries
  (A). Without the side information (the 2-d coordinates in A), PCA,
  NMF and EBMF with $K = 3$ factors cluster some points accurately,
  but many other points are not clustered accurately (B--D). By
  contrast, Spatial PCA \citep{shang_spatially_2022} and our new
  method, cEBMF, which both incorporate the side information into the
  prior, more accurately cluster the points (E, F). (For consistency
  of visualization, the top 2 PCs of the ${\bf L}$ matrices from NMF,
  EBMF and cEBMF are shown.) Spatial PCA assumes the data points are
  spatial, whereas cEBMF does make this assumption; instead, it has a flexible
  prior that is adapted to the data. This learned prior is shown in G:
  the color of the points shows
  %
  %
  the prior probability that row $i$, column $k$ of
  ${\bf L}$ is nonzero ({\color{SkyBlue} \bf blue} = low prior probability,
  {\color{OrangeRed} \bf red} = high prior probability). See
  Sections \ref{sec:cebmf-main} and \ref{sec:exp} for definitions and
  additional details.}
  %
  %
\end{figure}

In this paper, we describe a novel matrix factorization framework that
allows high-dimensional row and column data to guide the matrix
factorizations without having to make specific assumptions about how
these data inform the factorization. For example, although our
framework can be applied to data
%
%
that exhibit spatial properties, it does not assume or require that
the data be spatial. Our framework is also flexible in that it
includes many existing approaches as special cases, including
unconstrained matrix factorization \citep{wang_empirical_2021,
  zhong_empirical_2022}, non-negative matrix factorization
\citep{lin_optimization_2020}, semi-non-negative matrix factorization
\citep{ding_convex_2010}, and more recent methods that incorporate
side information \cite{wang_mfai_2024}.  These features are achieved
by taking an empirical Bayes approach, building on the recent
empirical Bayes matrix factorization (EBMF) framework
\citep{wang_empirical_2021, zhong_empirical_2022}. In particular, we
extend the EBMF approach of \cite{wang_empirical_2021} with
adaptive priors that are modified by the side information. We call
this approach ``covariate-moderated empirical Bayes matrix
factorization,'' or ``cEBMF'' for short. See
Fig.~\ref{fig:toy_example} for a toy example that illustrates the key
features of cEBMF.


\section{Related work}
\label{sec:related_work}

The literature on matrix factorization methods that incorporate side
information is quite extensive. The different methods make different
modeling assumptions, and are typically motivated by certain types of
data. Although it is not possible to review all relevant literature
here, we discuss a few of the most important or related methods.



Several variants of the topic model---which can be viewed as matrix
factorizations with ``sum-to-one'' constraints on ${\bf L}$ and ${\bf
  F}$ \citep{fastTopics}---incorporate side information in different
ways; for example, the correlated topic model
\citep{lafferty_correlated_2005} and the structural topic model
\citep{roberts_model_2016} incorporate document-level side information
into the priors on ${\bf L}$. Collective matrix factorization (CMF)
\citep{singh_relational_2008, vandyk2009group, bouchard2013convex} has
gained considerable interest, but CMF is based on ideas that are quite
different from cEBMF: like cEBMF, CMF assumes that the side
information is the form of a matrix, but unlike cEBMF, CMF assumes
that the side information factorizes in a similar way to ${\bf
  Z}$. Clearly, this assumption will not make sense for some types of
data. Another prominent theme in matrix factorization with side
information is incorporating group-level or categorical information,
including ontological data. Among the methods in this area are CTPF
\cite{gopalan2014contentbased} and the method of
\cite{hu2016nmf}. Another important class of methods related to cEBMF
from the deep learning literature are variational autoencoders (VAE)
\cite{kingma_auto-encoding_2014}, conditional variational autoencoders
(cVAE) \cite{sohn_learning_2015} and neural collaborative filtering
(NCF) \cite{he_neural_2017}. These methods generalize the concept of
matrix factorization to nonlinear embeddings.

The method that is most closely related to cEBMF is MFAI
\citep{wang_mfai_2024} (see also \cite{porteous_bayesian_2010-1} for
related ideas). MFAI is in fact a special case of cEBMF in which the
priors on ${\bf F}$ are normal and the prior means are informed by the
covariates. Like cEBMF, MFAI allows these priors to be adapted
separately for each dimension $k$. However, MFAI is not nearly as
general as cEBMF; it implements only a single model, a single prior
family with a specific parametric form, a specific procedure for
fitting these priors (using gradient boosted tree methods
\cite{friedman2001greedy}), and it only accommodates row-wise side
information.

%
%

Several matrix factorization methods have been developed specifically
for spatial transcriptomics data.
%
%
Spatial PCA \citep{shang_spatially_2022} models the spatial similarity
among rows of ${\bf L}$ using Gaussian process prior. (Spatial PCA is
similar to GP-LVM \citep{lawrence_hierarchical_2007}. See also
\citep{zou2012kernelized}.) An NMF version of this approach generates
``parts-based representations'' guided by the spatial context of the
data points \citep{townes_nonnegative_2023}. More recently, IRIS
\citep{iris} regularizes the matrix factors through a penalty function
that encodes the spatial information in a graph (see also
\cite{cai_graph_2011}).

\section{Covariate-moderated empirical Bayes matrix factorization}
\label{sec:cebmf-main}

\subsection{Background: empirical Bayes matrix factorization}

Empirical Bayes matrix factorization (EBMF)
\citep{wang_empirical_2021} is a flexible
framework for matrix factorization: it approximates a matrix ${\bf Z}
\in \mathbb{R}^{n \times p}$ as the product of two low-rank matrices,
\begin{equation}
{\bf Z} \approx {\bf L}{\bf F}^T,
\end{equation}
where ${\bf L} \in \mathbb{R}^{n \times K}$, ${\bf F} \in
\mathbb{R}^{p \times K}$, and $K \geq 1$. (In our applications, $K \ll
n$, $p$.) EBMF assumes a normal model of the data,
\begin{equation}
{\bf Z} = {\bf L}{\bf F}^T + {\bf E}, \quad
e_{ij} \sim N(0, \tau_{ij}^{-1}),
\label{eq:ebmf-likelihood}
\end{equation}
in which $N(\mu, \sigma^2)$ denotes the normal distribution with mean
$\mu$ and variance $\sigma^2$, and the residual variances
$\tau_{ij}^{-1}$ may vary by row ($i$) or by column ($j$) or
both. (EBMF, and by extension cEBMF, also allows ${\bf Z}$ to contain
missing values \cite{wang_empirical_2021}, which is important in many
applications of matrix factorization, including
%
%
collaborative filtering; see Sec.~\ref{sec:movielens}.) EBMF assumes
prior distributions for elements of ${\bf L}$ and ${\bf F}$, which are
themselves estimated among pre-specified prior families $
\mathcal{G}_{\ell,k}$ and $\mathcal{G}_{f,k}$, respectively:
\begin{equation}
\begin{aligned}
\ell_{ik} &\sim g_k^{(\ell)},
\quad g_k^{(\ell)}\in \mathcal G_{\ell,k} ,
\quad k = 1, \ldots, K \\
f_{jk} &\sim g_k^{(f)},
\quad g_k^{(f)}\in \mathcal G_{f,k},
\quad k = 1, \ldots, K.
\end{aligned}
\label{eq:ebmf-prior}
\end{equation}

The flexibility of EBMF comes from the wide range of possible prior
families (including non-parametric families) \citep{ebnm}.  Different
choices of prior family correspond to different existing matrix
factorization methods. For example, if all families $\mathcal
G_{\ell,k}$ and $\mathcal G_{f,k}$ are the family of zero-mean normal
priors, then ${\bf L} {\bf F}^T$ is similar to a truncated singular
value decomposition (SVD) \citep{Nakajima, nakajima2013}. When the
prior families are all point-normal (mixture of a point mass at zero
and a zero-centered normal), one obtains empirical Bayes versions of
sparse SVD or sparse factor analysis \citep{yang_sparse_2014,
  witten_penalized_2009, Engelhardt}. The prior families can also
constrain ${\bf L}$ and ${\bf F}$; for example, families that only
contain distributions with non-negative support result in empirical
Bayes versions of NMF. In summary, EBMF
(\ref{eq:ebmf-likelihood}--\ref{eq:ebmf-prior}) is a highly flexible
modeling framework for matrix factorization that includes important
previous methods as special cases, but also many new combinations
(e.g., \cite{liu2023dissecting}).

\subsection{The cEBMF modeling framework}

In covariate-moderated EBMF (cEBMF), we assume that we have some
``side information'' (covariates) for rows and/or columns of ${\bf Z}$
\citep{anndata, adams2010incorporating}. Let $\bx_i$ denote the
available information for the $i$th row of $\bZ$, and let
$\by_j$ denote the available information for the $j$th column of
$\bZ$. In principle, $\bx_i$ and $\by_j$ can be any information
processable by a neural net (text, graph, image, other structured
data), but for simplicity we assume that this information is stored as a
matrix. Therefore, let $\bX \in {\mathbb R}^{n \times n_x}$ be a
matrix containing information on the rows of $\bZ$, with $\bx_i$
corresponding to the $i$th row of $\bX$ (e.g., $\bx_i$ might
contain the 2-d coordinate of cell $i$).
%
%
Similarly, let $\bY \in \mathbb{R}^{p \times n_y}$ contain information
on the columns of $\bZ$, with $\by_j$ corresponding to the $j$th row
of $\bY$.
%
%
In cEBMF, we incorporate this side information into the model through
{\em parameterized priors},
\begin{equation}
\begin{aligned}
\ell_{ik} &\sim g_k^{(\ell)}({\bm x}_i),\quad
g_k^{(\ell)}({\bm x}_i) \in \mathcal G_{\ell,k},  \quad k = 1, \ldots, K \\
f_{jk} &\sim g_k^{(f)}({\bm y}_j),\quad
g_k^{(f)}({\bm y}_i) \in \mathcal G_{f,k},\quad k = 1, \ldots, K,
\end{aligned}
\label{eq:cebmf-prior}
\end{equation}
where $g_k^{(\ell)}({\bm x}_i)$ is a probability distribution within the
family $\mathcal{G}_{\ell,k} $, parameterized by ${\bm x}_i$, and
$g_k^{(f)}({\bm y}_j)$ is a probability distribution within the family
$\mathcal{G}_{f,k}$ parameterized by ${\bm y}_i$.

A key limitation of many existing approaches is that they integrate
the side information using restrictive parametric models that may or
may not be appropriate for the particular application. Another
limitation is that the priors chosen for these methods may make strong
or perhaps unrealistic assumptions about the structure underlying the
data; for example, Gaussian process priors, which have been used in
matrix factorization (e.g., \cite{shang_spatially_2022,
  lawrence_hierarchical_2007, lawrence2005probabilistic}), typically
assume that the factors vary smoothly in space, which makes it
difficult to capture sharp changes at boundaries
\citep{kim2005analyzing}. Existing methods also typically rely on
hyperparameters that need to be tuned or selected (e.g., using
cross-validation).

To address these issues, we propose cEBMF, a method that:
\begin{enumerate}[leftmargin=2em]

\item Can leverage a large variety of models (e.g., multinomial
  regression, multilayer perceptron, graphical neural nets,
  convolutional neural nets) to integrate the side information into
  the prior.

\item Can use families of priors that are flexible in form and
  thus do not make strong assumptions.

\item Allows automatic selection of the hyperparameters in
  \eqref{eq:cebmf-prior} via an empirical Bayes approach.

\end{enumerate}

More formally, we fit a prior for each column $k$ of ${\bf L}$, which
maps each vector of covariates $\bx_i$ to a given element
$g_k^{(\ell)}({\bm x}_i) \in \mathcal G_{\ell,k}$, and similarly for
each column $k$ of ${\bf F}$.
%
%
%
In Sec.~\ref{sec:cebmf-algorithm}, we describe a simple yet general
algorithm that simultaneously learns the factors ${\bf L}, {\bf F}$
and the priors $g_k^{(\ell)}({\bm x}_i), g_k^{(f)}({\bm y}_j)$. A
%
%
PyTorch-based \citep{pytorch} implementation of cEBMF with several
different parameterized prior families is available at
\url{https://github.com/william-denault/cebmf_torch}.

\subsubsection{An illustration: cEBMF with
  side information on factor sparsity}
\label{sec:cebmf-spike-and-slab}

Here we illustrate the implementation of the cEBMF framework using a
simple yet broadly applicable prior family.  This prior family assumes
that the covariates ${\bf X}, {\bf Y}$ {\em only inform the pattern of
  sparsity}---that is, the placement of zeros---in ${\bf L}$ and ${\bf
  F}$. This type of prior is of particular interest for matrix
factorization because matrix factorizations are typically invariant to
rescaling, and therefore priors that inform the magnitudes of
$\ell_{ik}$ and $f_{jk}$ are difficult to design. (By ``invariant to
rescaling,'' we mean that the likelihood or objective does not change
if we replace ${\bf L}{\bf F}^T$ by $\tilde{\bf L} \tilde{\bf F}^T$,
where $\tilde{\bf L} = {\bf L}{\bf }{\bf D}, \tilde{\bf F} = {\bf
  F}{\bf D}^{-1}$, and ${\bf D}$ is an invertible diagonal matrix.)
%
%
We define this prior family as
\begin{equation}
\label{eq:parameterized-spike-and-slab}
\mathcal{G}_{\mathrm{ss}}
\colonequals \{ g : g(u) = (1 - \pi({\bm x}, {\bm\theta})) \delta_0(u) +
\pi({\bm x}, {\bm\theta} ) g_1(u; {\bm\omega}) \},
\end{equation}
in which $\delta_0(u)$ denotes the point-mass at zero, $g_1(u;
{\bm\omega})$ denotes the density of some probability distribution
$g_1({\bm\omega})$ on $u \in \mathbb{R}$, and ${\bm x} \in
\mathbb{R}^m$ denotes the covariate. For example, when $g_1$ is the
normal distribution and ${\bm\omega}$ specifies the mean and variance,
\eqref{eq:parameterized-spike-and-slab} is a family of parameterized
``spike-and-slab'' priors \citep{spike-and-slab}, and cEBMF with
$\mathcal{G}_{\ell,k} = \mathcal{G}_{\text{ss}}, \mathcal{G}_{f,k} =
\mathcal{G}_{\text{ss}}$ implements a version of sparse factor
analysis \citep{yang_sparse_2014, witten_penalized_2009, Engelhardt}
in which the sparsity of the factors is informed by the
covariates. (Note that the ``ss'' in $\mathcal{G}_{\text{ss}}$ is
short for ``spike-and-slab.'')  Alternatively, if $g_1$ is a
distribution with support only on non-negative numbers, such as an
exponential distribution, then cEBMF implements a version of sparse
NMF. The free parameters are ${\bm\theta}$, which control the weight
on the ``spike'', $\delta_0$, and ${\bm\omega}$, which control the
shape of the ``slab'', $g_1$.
%
%
One simple parameterization of $\pi({\bm x}, {\bm\theta})$ uses a
logistic regression model,
\begin{equation}
\pi({\bm x}, {\bm\theta}) = \phi\big({\textstyle \theta_0 +
\sum_{t=1}^m x_t\theta_t} \big),
\end{equation}
where $\phi(x) \colonequals 1/(1 + e^{-x})$ denotes the sigmoid
function, and ${\bm\theta} \in {\mathbb R}^{m + 1}$. Most of the
parameterized prior families used in this paper and in the cEBMF
software are variants or elaborations on $\mathcal{G}_{\mathrm{ss}}$.

%
%
%

\subsection{The cEBMF learning algorithm}
\label{sec:cebmf-algorithm}

%

A key feature of the cEBMF modeling framework is that the algorithm
for fitting the priors and estimating the factorization is simple to
describe and
%
%
often straightforward to implement.
%
%
In brief, the cEBMF learning algorithm reduces a complex model-fitting
task to a series of simpler subproblems. Each of these subproblems
involves fitting a covariate-moderated variant of an empirical Bayes
normal means (EBNM) model \citep{ebnm}.  This also has the advantage
of making the cEBMF framework and software modular so that a method
that solves a covariate-moderated EBNM problem can be ``plugged in''
to the generic cEBMF algorithm.

\subsubsection{Background: empirical Bayes normal means}
\label{sec:ebnm}

Given observations $\hat{\beta}_i \in \mathbb{R}$ with known standard
deviations $s_i > 0$, $i = 1, \dots, n$, the normal means model
\citep{Robbins51, efron1972limiting, stephens_false_2017} is
\begin{equation}
\hat{\beta}_i \overset{\text{ind.}}{\sim} N(\beta_i, s_i^2),
\label{eqn:nm_problem}
\end{equation}
in which the ``true'' means $\beta_i \in \mathbb{R}$ are unknown.
It is further assumed that the unknown means are 
\begin{equation}
\label{eqn:nm_prior}
\beta_i \overset{\text{i.i.d.}}{\sim} g \in \mathcal{G},
\end{equation}
where $\mathcal{G}$ is some pre-specified family of probability
distributions.

The empirical Bayes (EB) approach to fitting the normal means model
(\ref{eqn:nm_problem}--\ref{eqn:nm_prior}) exploits the fact that the
noisy observations $\hat{\beta}_i$ contain information not only about
the underlying means $\beta_i$ but also how the means are collectively
distributed.
%
%
EB ``borrows information'' across the observations to estimate $g$;
typically this is done by maximizing the marginal log-likelihood of
(\ref{eqn:nm_problem}--\ref{eqn:nm_prior}). The unknown means
$\beta_i$ are typically estimated by their posterior means (given the
estimate of $g$).

To adapt the EBNM model (\ref{eqn:nm_problem}--\ref{eqn:nm_prior}) to
cEBMF, we allow the prior for the $i$th unknown mean to
depend on additional data ${\bm d}_i$ and parameters ${\bm\theta}$,
\begin{equation}
\beta_i \overset{\text{ind.}}{\sim}
g({\bm d}_i,\bm\theta ) \in \mathcal{G},
\label{eq:cebnm_prior}
\end{equation}
so that each combination of $\bm\theta$ and ${\bm d}_i$ maps to an
element of $\mathcal{G}$. We refer to this as ``covariate-moderated
EBNM'' (cEBNM). Solving the cEBNM problem therefore involves two key
computations:

\paragraph{1. Estimate the model parameters.} 
 
Compute
\begin{equation}
\hat{\bm\theta} \colonequals
\argmax_{{\bm\theta} \,\in\, {\bf R}^m} \mathcal{L}({\bm\theta}),
\label{eqn:argmax_theta}
\end{equation}
where $\mathcal{L}({\bm\theta})$ denotes the marginal likelihood,
\begin{equation} 
\label{eqn:lik}
\mathcal{L}({\bm\theta}) \colonequals
p(\hat{\bm \beta} \mid {\bm s}, {\bm \theta},
{\bf D}) = 
\prod_{i=1}^n \textstyle 
\int N(\hat{\beta}_i; \beta_i, s_i^2) \,
g(\beta_i; {\bm d}_i, {\bm\theta}) \, d\beta_i,
\end{equation}
in which $\hat{\bm\beta} = (\hat{\beta}_1, \ldots, \hat{\beta}_n)$,
${\bm s} = (s_1, \ldots, s_n)$, ${\bf D}$ is a matrix storing
%
%
${\bm d}_1, \ldots, {\bm d}_n$, $N(\hat{\beta}_i; \beta_i, s_i^2)$
denotes the density of $N(\beta_i, s_i^2)$ at $\hat{\beta}_i$, and
$g(\beta_i; {\bm d}_i, {\bm\theta})$ denotes the density of $g( {\bm
  d}_i, {\bm\theta})$ at $\beta_i$.
 
\paragraph{2. Compute posterior summaries.}

Compute summaries from the posterior distributions given
the estimated prior,
%
%
\begin{equation} 
p(\beta_i \mid \hat{\beta}_i, s_i, \hat{\bm\theta}, {\bf D}) = 
\frac{N(\hat{\beta}_i; \beta_i, s_i^2) \,
  g(\beta_i; {\bm d}_i, \hat{\bm\theta})}
{\int N(\hat{\beta}_i; t, s_i^2) \, g(t; {\bm d}_i, \hat{\bm\theta}) \, dt}.
\label{eqn:ebnm_post}
\end{equation}

For many classical prior families, such as the spike-and-slab family
in Sec.~\ref{sec:cebmf-spike-and-slab}, the integrals in
\eqref{eqn:lik} and \eqref{eqn:ebnm_post} can be computed
analytically. More generally, standard numerical techniques such as
Gauss-Hermite quadrature may provide reasonably fast and accurate
solutions for prior families that do not result in closed-form
integrals since the integrals in \eqref{eqn:lik} and
\eqref{eqn:ebnm_post} are one-dimensional. As a result,
$\hat{\bm\theta}$ can often be obtained efficiently using
off-the-shelf optimization algorithms even when the chosen priors do
not result in analytical integrals.
%
%

In summary, solving the cEBNM problem consists of finding a mapping
from the known quantities $(\hat{\bm\beta}, {\bm s}, {\bf D})$ to a
tuple $(\hat{\bm\theta}, \hat{q})$, where each $({\bm
  d}_i,\hat{\bm\theta} )$ maps to an element $g({\bm d}_i,
\hat{\bm\theta}) \in \mathcal{G}$, and $\hat{q}$ is the posterior
distribution of the unknown means, $\hat{q}({\bm\beta}) \colonequals
\prod_{i=1}^n p(\beta_i \mid \hat{\beta}_i, s_i, \hat{\bm\theta}, {\bf
  D})$. To facilitate the description of the cEBMF algorithm below, we
denote this mapping as
\begin{equation} 
\mathrm{cEBNM}(\hat{\bm\beta}, {\bm s}, {\bf D}, \mathcal{G}) =
(\hat{\bm\theta}, \hat{q}).
\label{eq:cEBNM_mapping}
\end{equation}
Note that in practice the full posterior $\hat{q}({\bm\beta})$ is not
needed; the first and second posterior moments are sufficient (see
Sec.~\ref{sec:alg_sketch}). Any prior family is admissible under the
cEBMF framework so long as the mapping \eqref{eq:cEBNM_mapping} is
computable (either numerically or analytically).
%
%

\subsubsection{Algorithm}
\label{sec:alg_sketch}

Given a method for solving the cEBNM problem (Sec.~\ref{sec:ebnm}),
the cEBMF model can be fitted using a simple coordinate ascent
algorithm. In brief, the cEBMF algorithm maximizes an objective
function---the evidence lower bound (ELBO)
\citep{blei_variational_2017} under a variational approximation with
conditional independence assumptions on ${\bf L}$ and ${\bf F}$ (see
the Appendix)---by iterating over the following updates for each
factor $k = 1, \ldots, K$ until some stopping criterion is met:
\begin{enumerate}

\item Disregard the $k$th factor in $\bar{\bf R}$, the $n \times p$
  matrix expected residuals, $\bar{\bf R}^k = \bar{\bf R} +
  \bar{\bm\ell}_k \bar{\bm f}_k^T$.

\item For each $i = 1, \ldots, n$, compute the least-squares estimates
  of $\ell_{ik}$, denoted $\hat{\ell}_{ik}$, and the standard
  deviations $s_{ik}^{\ell}$ of these estimates,
\begin{align}
\hat{\ell}_{ik} &= \textstyle 
(s_{ik}^{\ell})^2 \sum_{j=1}^p \tau_{ij} \bar{r}_{ij}^k \bar{f}_{jk} \\
s_{ik}^{\ell} &= \textstyle (\sum_{j=1}^p \tau_{ij} \bar{f}_{jk}^2)^{-1/2},
\end{align}
where $\bar{f}_{jk}$ and $\bar{f}_{jk}^2$ denote, respectively, the
first and second posterior moments of $f_{jk}$.

\item Update $g_k^{(\ell)} \in \mathcal{G}_{\ell,k}$ by solving
  \eqref{eqn:argmax_theta}, in which we make the following
  substitutions in \eqref{eqn:argmax_theta}: $\hat{\beta}_i \leftarrow
  \hat{\ell}_{ik}, s_i \leftarrow s_{ik}^{\ell}, i = 1, \ldots, n,
      {\bf D} \leftarrow {\bf X}, \mathcal{G} \leftarrow
      \mathcal{G}_{\ell, k}$.

\item Making the same substitutions in \eqref{eqn:ebnm_post}, update
  the posterior means $\bar{\bm\ell}_k = (\bar{\ell}_{1k}, \ldots,
  \bar{\ell}_{nk})^T$ and posterior second moments $\bar{\bm\ell}_k^2 =
  (\bar{\ell}_{1k}^2, \ldots, \bar{\ell}_{nk}^2)^T$.
  
\item Perform updates similar to those in Steps 2--4 to update
  $\bar{\bm f}_k$, $\bar{\bm f}_k^2$ and $g_k^{(f)} \in \mathcal{G}_{f,k}$.

\item Update the matrix of expected residuals,
$\bar{\bf R} = \bar{\bf R}^k - \bar{\bm\ell}_k \bar{\bm f}_k^T$.

\end{enumerate}
These steps are iterated until some stopping criterion is met. The
algorithm must be initialized with initial estimates of $\bar{\bf L},
\bar{\bf F}$. The expected residuals are then initialized as $\bar{\bf
  R} = {\bf Z} - \bar{\bf L} \bar{\bf F}^T$. Note that to simplify
presentation we have omitted some details such as how to update the
residual variances $\tau_{ij}^{-1}$. These and other details are
provided in the Appendix.
%
%


\section{Experiments}
\label{sec:exp}

%

\subsection{Simulations}
\label{sec:simulations}

To assess the benefits of cEBMF, we compared cEBMF with other matrix
factorization methods in simulated data sets. We compared with several
methods that do not use side information, including EBMF (flashier R
package \citep{wang_empirical_2021, ebnm}), penalized matrix
decomposition (``PMD''; PMA R package \citep{witten_penalized_2009}),
and a variational autoencoder (VAE) \citep{kingma_auto-encoding_2014}
implemented in PyTorch \citep{pytorch}. We also compared with other
methods that use side information, including MFAI (mfair R package
\citep{wang_mfai_2024}), Spatial PCA \citep{shang_spatially_2022},
conditional VAE (cVAE) \cite{sohn_learning_2015}, and neural
collaborative filtering (NCF) \cite{he_neural_2017}. cVAE and NCF were
also implemented in PyTorch. Note that Spatial PCA accepts only a
specific type of side information, the 2-d coordinates of the data
points, so was not included in all the simulations.

We compared the methods in four simulation scenarios designed to
capture a range of settings where one might perform a matrix
factorization analysis, with or without side information: (1) a
``sparsity-driven covariate'' setting in which the covariates only
informed the sparsity pattern of ${\bf L}$ and ${\bf F}$; (2) an
``uninformative covariate'' setting in which the covariates provided
no information about the true matrix factorization; (3) a
``tiled-clustering'' setting in which ${\bf L}$ depended on the 2-d
location of the data points; and (4) a ``shifted tiled-clustering''
setting in which the cEBMF priors were unable recover the true data
generating process. The latter scenario was used to assess cEBMF under
model misspecification. We simulated 100 data sets in each setting,
and we assessed the ability of each method to recover the true matrix
factorization as measured by root mean squared error (RMSE) between
the true matrix factorization ${\bf L}{\bf F}^T$ and estimated matrix
factorization $\hat{\bf L} \hat{\bf F}^T$. More detailed descriptions
of the simulations and the methods compared are given in the Appendix.

\begin{wrapfigure}{r}{0.65\textwidth}
\vspace*{-1em}
\includegraphics[width=0.65\textwidth]{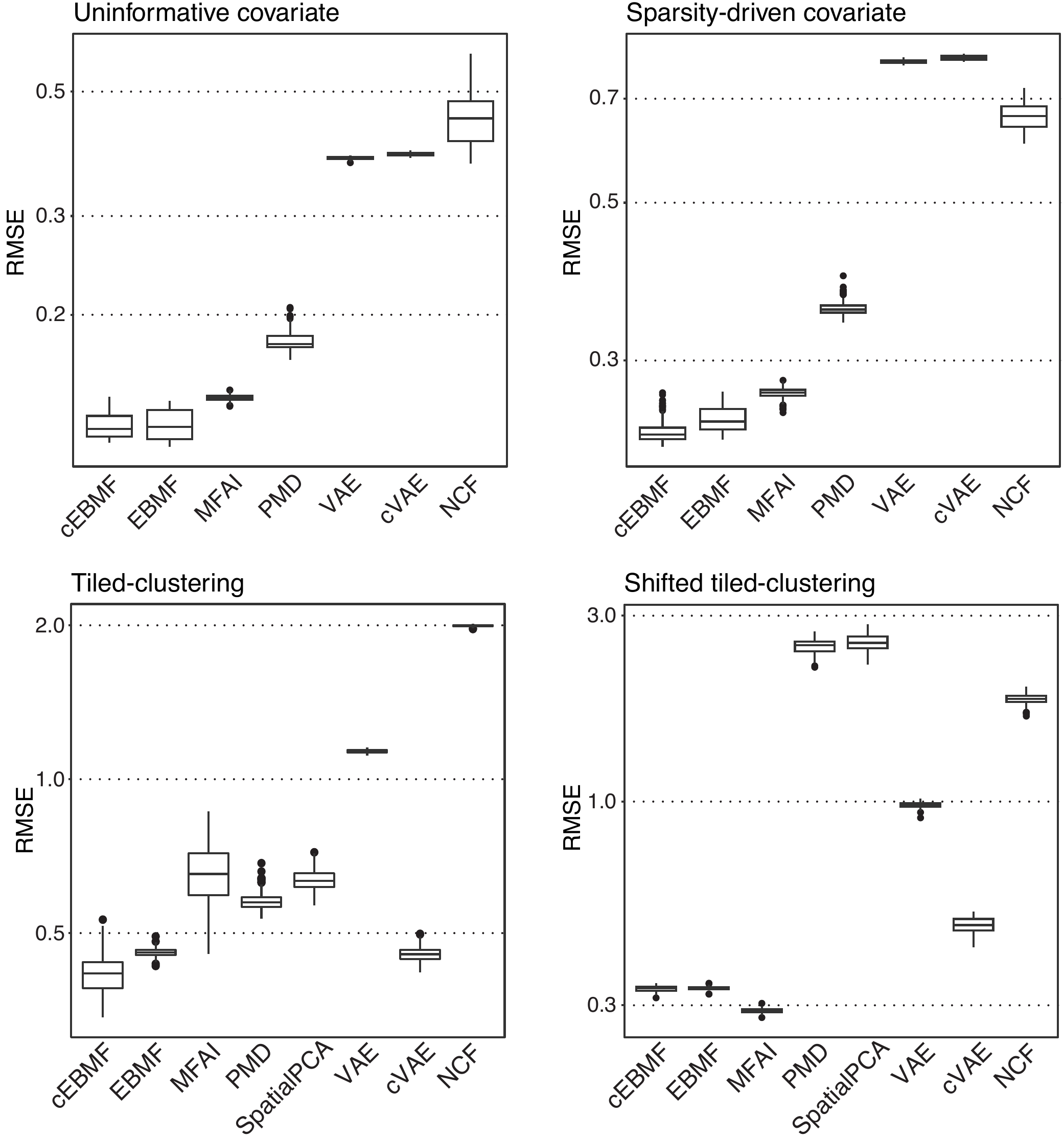}
\vspace*{-1.2em}
\caption{Performance of the different matrix factorization methods in
  the simulated data sets. Each boxplot summarizes the root mean
  squared errors (RMSEs) across 100 simulations in that scenario.
  (Lower RMSEs are better.) See Figures
  \ref{fig:uninformative-sims}--\ref{fig:shifted-tiled-sims} in the
  Appendix for additional results from the simulations. Note
  Fig.~\ref{fig:toy_example} shows results from one of the
  tiled-clustering simulations in detail.}
\label{fig:simulation_res1}
\vspace*{-1.25ex}
\end{wrapfigure}

The results are summarized in Fig.~\ref{fig:simulation_res1}. cEBMF
was generally more accurate than the other methods, particularly when
the covariates were informative; cEBMF achieved the greatest gains
over EBMF in the tiled-clustering setting where the covariates were
also the most informative. Reassuringly, cEBMF did not perform worse
than EBMF in settings with an uninformative covariate or a prior that
was misspecified (``shifted tiled-clustering''). The deep learning
approaches (VAE, cVAE, NCF) generally performed worse than the other
methods. cVAE sometimes outperformed VAE when the side information was
highly informative, such as in the tiled-clustering scenario, but did
not provide improvements over VAE in the more challenging
sparsity-driven scenario. We also ran Spatial PCA on the
tiled-clustering and shifted tiled-clustering data sets where the
factors were partly driven by the 2-d locations of the data
points. Despite the fact that Spatial PCA can exploit the side
information, it had worse accuracy than EBMF which did not use the
side information. This may be because Spatial PCA makes assumptions
(e.g., orthogonal factors) that were not met by our simulations.  MFAI
generally performed worse than EBMF and cEBMF except in the shifted
tiled-clustering setting; MFAI is a much less flexible model than
cEBMF and therefore its performance was sensitive to the
appropriateness of its modeling assumptions. (All models were
misspecified in the shifted tiled-clustering setting, but perhaps MFAI
was the least misspecified.) Additional results including comparisons
with other methods (PCA/SVD, Sparse SVD \cite{yang_sparse_2014}, CMF
\citep{singh_relational_2008}) are in the Appendix.


\subsection{Collaborative filtering}
\label{sec:movielens}

\begin{figure}[t]
\centering 
\includegraphics[width=0.9\textwidth]{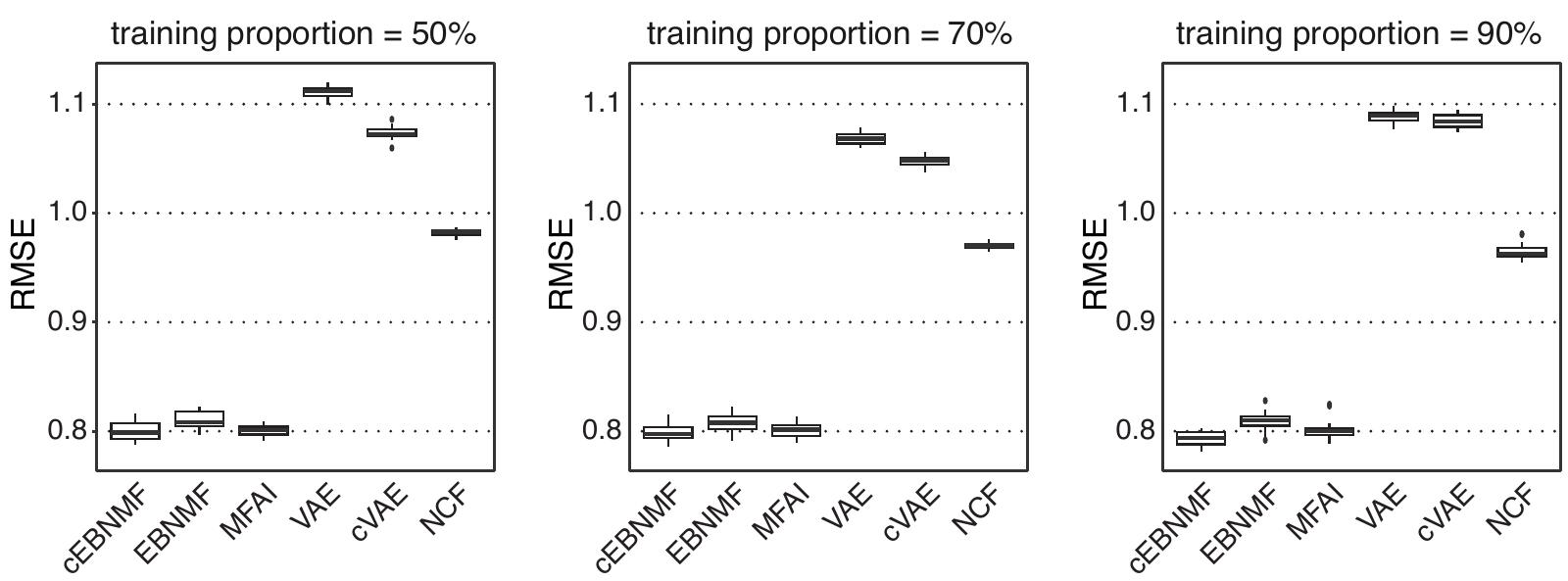}
\vspace*{-1ex}
\caption{Prediction performance of different matrix factorization and
  deep learning methods in the MovieLens 100K data
  \citep{harper_movielens_2015}.
  %
  %
  Training proportion = $X\%$ means that $X\%$ of the movie ratings were
  used in training, and the remaining $(100 - X)\%$ were used
  to evaluate accuracy (measured using RMSE). The results at each
  training proportion are from 10 random training-test splits.}
\label{fig:movielnes_res}
\end{figure}

To provide a quantitative assessment of the matrix factorization
methods in real data, we ran the same methods on the MovieLens 100K
data \citep{harper_movielens_2015}, a standard collaborative filtering
benchmark in which the goal is to predict the unobserved elements of
the matrix.
%
%
Here, ${\bf Z}$ was a $\mbox{1,682} \times 943$ matrix containing
integer-valued movie ratings, with rows corresponding to movies and
columns corresponding to users. Since most (93\%) of the movie ratings
were missing, this example highlights the ability of these methods
to handle missing data (unlike most NMF methods).  The side
information ${\bf X}$ was a $\mbox{1,682} \times 19$ binary matrix
containing information about the movie's genre (comedy, adventure,
etc).  We held out some of the moving ratings at random, and used
these held-out ratings as a test set.
%
%
%

We ran EBMF and cEBMF so as to produce non-negative matrix
factorizations, which is common in collaborative filtering (e.g.,
\citep{singh_relational_2008}). Therefore, in the results we labeled
these methods as ``EBNMF'' and ``cEBNMF''. (Note that MFAI cannot
produce non-negative matrix factorizations.) To enforce non-negativity
in ${\bf L}$ and ${\bf F}$, we used mixture-of-exponentials
priors. For cEBNMF, the side information was incorporated into the
priors on ${\bf L}$ using a multi-layer perceptron. (See the Appendix
for further details.)
%
%
Since we didn't know the true number of factors, $K$ was chosen
adaptively in EBMF, cEBMF and MFAI. (We set an upper limit of 7 for
EBMF and cEBMF, and 12 for MFAI.)

The results are summarized in Fig.~\ref{fig:movielnes_res}.  Both MFAI
and cEBNMF were able to use the side information (the movie genres) to
improve over EBNMF, and all three matrix factorization methods were
more accurate than the deep learning methods. We conjecture that the
deep learning methods would have performed better with more data
(such as the more recently released MovieLens data sets that are much
larger). On the MovieLens 100K data, cEBNMF yielded overall the best
prediction accuracy across the different training-set splits.




\subsection{Spatial transcriptomics}
\label{sec:spatial-transcriptomics}

\begin{figure}[t]
\centering 
\includegraphics[width=0.915\textwidth]{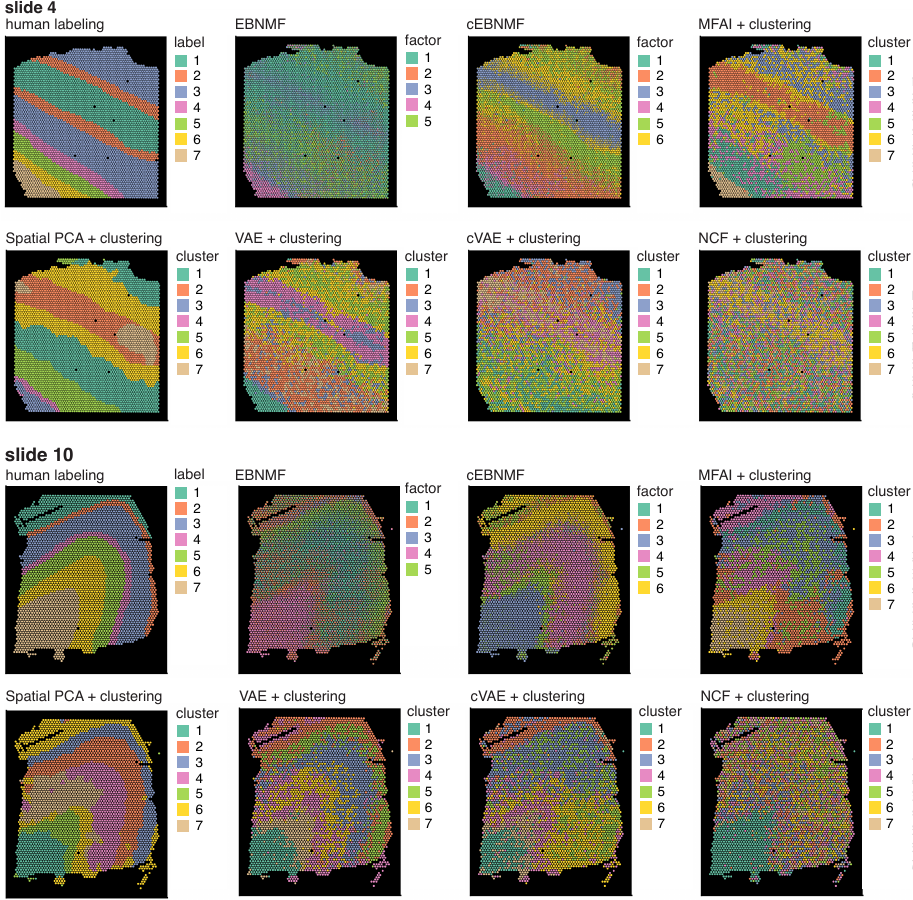}
\caption{Results on slides 4 (top row) and 10 (bottom row) of the
  DLPFC spatial transcriptomics data \citep{pardo_spatiallibd_2022,
    maynard_transcriptome-scale_2021}. For the NMF, EBMF and cEBMF
  results, each data point (``pixel'') $i$ is shown as a pie chart
  using the relative values of $i$th row of ${\bf L}$ (after
  performing an ``LDA-style'' post-processing of ${\bf L}, {\bf F}$
  \citep{townes_nonnegative_2023}). (Higher-resolution versions of
  these images are available online at
  \url{https://github.com/william-denault/cEBMF_experiment}.) Since
  MFAI and other methods did not produce a non-negative matrix
  factorization, we clustered the low-dimensional embeddings using the
  same approach that was used in Spatial PCA
  \citep{shang_spatially_2022}.
  %
  %
  CMF results for these two slices and additional results on
  all 12 slices are given in the Appendix.}
  %
  %
\label{fig:spatial_soft_clustering}
\end{figure}

Although cEBMF was not specifically designed for spatial data, here we
show that cEBMF also yields compelling results from spatial
transcriptomics data \citep{marx_method_2021} by exploiting the side
information, the spatial locations of the data points. We illustrate
this using a data set \citep{pardo_spatiallibd_2022,
  maynard_transcriptome-scale_2021} that has been annotated by domain
experts and has been used in several papers to benchmark methods for
spatial transcriptomics (e.g., \citep{shang_spatially_2022,
  varrone_cellcharter_2024, zhao_spatial_2021, zhu_srtsim_2023}). The
data were collected from 12 slices of the human dorsolateral
prefrontal cortex (DLPFC). After data preprocessing, each slice
contained about 4,000 pixels and expression measured in about 5,000
genes ($n \approx \mbox{4000}$, $p \approx \mbox{5,000}$).\footnote{We
used the data prepared by the authors of the SpatialLIBD package
\citep{pardo_spatiallibd_2022} which were made available for download
at \url{https://research.libd.org/spatialLIBD/}.}

In this application, our aim was to generate a ``parts-based''
representation of the data, with the hopes that the ``parts'' would
resolve to biologically interpretable units (e.g., cell types, tissue
regions, gene programs) \citep{townes_nonnegative_2023, gomde,
  dey_visualizing_2017}. This is a fundamentally different aim from
the previous examples: in the previous examples, the goal was to
generate accurate low-dimensional representations, but we did not ask
whether the {\em individual dimensions} were accurate or
interpretable. With this aim in mind, we ran cEBMF so as to produce
non-negative matrix factorizations (``cEBNMF'') using the same priors
that were used for the MovieLens data. We compared to two other
non-negative matrix factorizations that did not leverage the side
information: NMF implemented in the R package NNLM
\citep{lin_optimization_2020}, and EBMF with point-exponential priors
(``EBNMF''). (The point-exponential prior is a simplification of the
mixture-of-exponentials prior with a single exponential component in
which the rate parameter is also learned.). We also compared to
several of the methods that were considered in the previous
experiments, including methods such as Spatial PCA and cVAE that make
use of the side information, and others that do not.

Spatial PCA deserves special mention because it was specifically
designed for spatial transcriptomics data
\citep{shang_spatially_2022}. Although Spatial PCA does not produce a
parts-based decomposition, the Spatial PCA software automatically
clusters the data points after projection onto the principal
components (PCs), and this clustering can be compared to the
non-negative matrix factorizations. Following
\cite{shang_spatially_2022}, we computed the top 20 PCs, then we ran
the walk-trap clustering algorithm \citep{pons_computing_2005} on the
PCs. Additional details of the Spatial PCA analysis and the other
methods are given in the Appendix.
%
%

Figure \ref{fig:spatial_soft_clustering} shows results on two of the
slices, with additional results on all 12 slices provided in the
Appendix. The manual annotations on the left-hand side should be
viewed as a useful reference point, but not necessarily the ``ground
truth''. (Consider that the data-driven annotations might identify
previously unknown cellular structures.) EBNMF, cEBNMF and MFAI
adapted the number of factors $K$ to the data (with upper limits of
50, 20 and 9, respectively). For the other methods, the number of
clusters was set to match the manual annotation. Qualitatively, some
of the factors from NMF and EBMF seem to correspond to the
expert-labeled regions, but several other factors appear to be
capturing other substructures that have no obvious spatial
quality. Comparatively, the cEBNMF results in slices 4 and 9 capture
the expert labeling much more closely, with most factors showing a
clear spatial quality. The clusters from Spatial PCA, MFAI and VAE
also capture spatial structure and expert labeling well, but with some
notable exceptions, e.g., Spatial PCA cluster 7 in slices 4 and
9. (The Spatial PCA software performed additional post-processing on
the clusters which is why these clusters look less ``noisy'' than the
others.) cVAE, despite using the side information, did not seem to
improve over VAE. The CMF results were comparatively poor
(Fig.~\ref{fig:sprna_mfai_cmf} in the Appendix), reflecting the
inappropriateness of the CMF assumptions in this setting. Note that
the NMF methods can capture continuous variation in expression within
and across cell types or regions---as well as the expectation that
some pixels might reflect combinations of cellular
structures---whereas the clustering cannot.

\subsection{Scalability benchmark}
\label{sec:benchmark}

\begin{table}[t]
\centering
\begin{tabular}{llrrrr}
\toprule
& & \multicolumn{4}{c}{number of rows ($n$)} \\
\cmidrule(lr){3-6} 
method & software & $10^3$ & $10^4$ & $10^5$ & $10^6$ \\
\midrule
EBMF & flashier \citep{wang_empirical_2021, ebnm}    
& 0.8 & 2.5 & 36.9 & 165.1 \\
cEBMF & {\tt cebmf\_torch}* & 5.2 & 35.5 & 416.3 & 3,403.6 \\
MFAI & mfair \citep{wang_mfai_2024}
 & 45.4 & 251.3 & 11,293.2 & -- \\
Spatial PCA & SpatialPCA \citep{shang_spatially_2022} & 
234.8 & 8,213.7 & -- & -- \\
\bottomrule
\end{tabular}
\vspace*{1ex}
\caption{Running times of matrix factorization methods on data sets in
  which we varied $n$, the number of rows in ${\bf X}$ and ${\bf Z}$.
  The numbers in the table are the average running times (in seconds)
  from 10 simulated data sets.  *Available at
  \url{https://github.com/william-denault/cebmf_torch}.}
\label{tab:runtime_comparison}
\end{table}

cEBMF can also handle much larger data sets than the MovieLens and
DLPFC data sets considered above. (One reason we did not use larger
data sets was to allow for comparison with methods that do not scale
well to large data sets.) To illustrate this, we ran EBMF and cEBMF on
``tiled-clustering'' data sets (using the same priors described
Sec. \ref{sec:simulations}), in which the data sets were simulated
with different numbers of rows, $n$. We compared the running times
with two other matrix factorization methods that make use of the side
information, MFAI and Spatial PCA (Table
\ref{tab:runtime_comparison}). While cEBMF had considerably higher
running times than EBMF on the same data, cEBMF completed in much less
time on average than MFAI and Spatial PCA. Further, while cEBMF was
able to handle data sets with 1 million rows, Spatial PCA struggled to
analyze data sets with 100,000 or more rows due to its high memory
usage; for example, Spatial PCA needed approximately 300 GB of memory
for $n = \mbox{100,000}$. MFAI crashed frequently in data sets with $n
\geq \mbox{100,000 rows}$ (only 2 of 10 of the runs completed at $n =
\mbox{100,000}$). Note this benchmark was performed on a computer with
32 GB memory, an NVIDIA GeForce RTX\textsuperscript{TM} 4070 GPU and
an AMD Ryzen\textsuperscript{TM} 9 7940HS CPU (8 cores, 16
threads). The EBMF and cEBMF algorithms were run for at most 20
iterations. See also the Appendix where we describe some of the
computational properties of the cEBMF algorithmic framework.

\section{Conclusions}

We have introduced cEBMF, a general and flexible framework for matrix
factorization that (i) incorporates side information through flexible
covariate-dependent priors and (ii) learns these priors from the data
using empirical Bayes ideas. Considerable effort has gone into
optimizing the software implementation building on our previous work
on this topic \citep{wang_empirical_2021, ebnm}. As a result, cEBMF
scales well to large data sets with, say, hundreds of thousands or
millions of rows and/or columns. Our experiments highlight the
importance of using matrix factorization models that make appropriate
assumptions about the data or are sufficiently flexible to adapt to
the data. In our experiments, cEBMF performed competitively against
other matrix factorization methods and deep learning approaches that
make use of the side information. Because the priors in cEBMF can take
the form of virtually any probabilistic model optimized via
equations~(\ref{eqn:argmax_theta}--\ref{eqn:lik}), our framework opens
the door to incorporating other types of side information, including
images and graphs.
  
{\bf Note:} R and Python code implementing the experiments is
available at \url{https://github.com/william-denault/cEBMF_experiment},
and a PyTorch-based implementation of cEBMF is available at at
\url{https://github.com/william-denault/cebmf_torch}.
%
%

\section*{Acknowledgments}

We thank the staff at the Research Computing Center at the University
of Chicago for providing the high-performance computing resources used
to implement the numerical experiments.  We also thank Deaglan
Bartlett and Augustin Marignier for their help in developing the
PyTorch implementation.  This work was supported in part by NIH grant
R01HG002585 (to M.S.)  and Eric and Wendy Schmidt AI in Science
Postdoctoral Fellowship, a Schmidt Sciences, LLC program (to
W.R.P.D.). Additional support came from the University of Chicago Data
Science Institute through the 2024 AI+Science Research Initiative. We
also thank the reviewers for their feedback.

\setlength{\bibsep}{4pt plus 1pt}
\bibliography{library}
\bibliographystyle{unsrtnat}

%
%
\clearpage

\section*{NeurIPS Paper Checklist}
\begin{enumerate}

\item {\bf Claims}

\item[] Question: Do the main claims made in the abstract and introduction accurately reflect the paper's contributions and scope?
    \item[] Answer: \answerYes{}
    \item[] Justification: The claims made in the abstract and introduction are supported by (i) a qualitative comparison of cEBMF to related work and (ii) empirical assessments in a variety of data sets.

\item {\bf Limitations}
    \item[] Question: Does the paper discuss the limitations of the work performed by the authors?
    \item[] Answer: \answerYes{}
    \item[] Justification: We empirically assessed the performance of cEBMF in the situation where the prior was misspecified.

\item {\bf Theory assumptions and proofs}
    \item[] Question: For each theoretical result, does the paper provide the full set of assumptions and a complete (and correct) proof?
    \item[] Answer: \answerYes{}
    \item[] Justification: The key theoretical results are found in the Appendix. We have provided proofs of Proposition \ref{prop_connection} and Lemma \ref{lemma:cebnm}, and we have clearly stated the assumptions.

    \item {\bf Experimental result reproducibility}
    \item[] Question: Does the paper fully disclose all the information needed to reproduce the main experimental results of the paper to the extent that it affects the main claims and/or conclusions of the paper (regardless of whether the code and data are provided or not)?
    \item[] Answer: \answerYes{}
    \item[] Justification: We have provided detailed descriptions of the experiments, including the software used. We have also provided the code that was used to generate all the results in the paper.

\item {\bf Open access to data and code}
    \item[] Question: Does the paper provide open access to the data and code, with sufficient instructions to faithfully reproduce the main experimental results, as described in supplemental material?
    \item[] Answer: \answerYes{}
    \item[] Justification: We have provided the R and Python code implementing the methods and experiments. The data sets were either included or instructions for accessing the data sets were given.

\item {\bf Experimental setting/details}
    \item[] Question: Does the paper specify all the training and test details (e.g., data splits, hyperparameters, how they were chosen, type of optimizer, etc.) necessary to understand the results?
    \item[] Answer: \answerYes{}
    \item[] Justification: Many of these details are given in the Appendix. Additional implementation details can

\item {\bf Experiment statistical significance}
    \item[] Question: Does the paper report error bars suitably and correctly defined or other appropriate information about the statistical significance of the experiments?
    \item[] Answer: \answerYes{}
    \item[] Justification: We summarize some results using boxplots with the conventional definitions for whiskers, box bounds, center line and outliers.

\item {\bf Experiments compute resources}
    \item[] Question: For each experiment, does the paper provide sufficient information on the computer resources (type of compute workers, memory, time of execution) needed to reproduce the experiments?
    \item[] Answer: \answerYes{}
    \item[] Justification: Details of the computing hardware used in the scability benchmark are given in Sec. \ref{sec:benchmark}.
    
\item {\bf Code of ethics}
    \item[] Question: Does the research conducted in the paper conform, in every respect, with the NeurIPS Code of Ethics \url{https://neurips.cc/public/EthicsGuidelines}?
    \item[] Answer: \answerYes{}
    \item[] Justification: We have reviewed the NeurIPS Code of Ethics and our work abides by it to the best of our knowledge.

\item {\bf Broader impacts}
    \item[] Question: Does the paper discuss both potential positive societal impacts and negative societal impacts of the work performed?
    \item[] Answer: \answerNA{}
    \item[] Justification: We don't believe that our work has any obvious 
direct social impacts.
    
\item {\bf Safeguards}
    \item[] Question: Does the paper describe safeguards that have been put in place for responsible release of data or models that have a high risk for misuse (e.g., pretrained language models, image generators, or scraped datasets)?
    \item[] Answer: \answerNA{}
    \item[] Justification: This work did not release data or models that would be considered to have a strong potential for misuse.

\item {\bf Licenses for existing assets}
    \item[] Question: Are the creators or original owners of assets (e.g., code, data, models), used in the paper, properly credited and are the license and terms of use explicitly mentioned and properly respected?
    \item[] Answer: \answerYes{}
    \item[] Justification: All the original data sources are credited/cited.

\item {\bf New assets}
    \item[] Question: Are new assets introduced in the paper well documented and is the documentation provided alongside the assets?
    \item[] Answer: \answerYes{}
    \item[] Justification: We provide the source code for our new methods, and it is accompanied with detailed documentation.

\item {\bf Crowdsourcing and research with human subjects}
    \item[] Question: For crowdsourcing experiments and research with human subjects, does the paper include the full text of instructions given to participants and screenshots, if applicable, as well as details about compensation (if any)? 
    \item[] Answer: \answerNA{}
    \item[] Justification: This work does not involve crowdsourcing or
research with human subjects.

\item {\bf Institutional review board (IRB) approvals or equivalent for research with human subjects}
    \item[] Question: Does the paper describe potential risks incurred by study participants, whether such risks were disclosed to the subjects, and whether Institutional Review Board (IRB) approvals (or an equivalent approval/review based on the requirements of your country or institution) were obtained?
    \item[] Answer: \answerNA{}
    \item[] Justification: this work does not involve research with human subjects.

\item {\bf Declaration of LLM usage}
    \item[] Question: Does the paper describe the usage of LLMs if it is an important, original, or non-standard component of the core methods in this research? Note that if the LLM is used only for writing, editing, or formatting purposes and does not impact the core methodology, scientific rigorousness, or originality of the research, declaration is not required.
    \item[] Answer: \answerNA{} 
    \item[] Justification: LLMs were not were not used in this work.

\end{enumerate}

\clearpage

\appendix

{\Large\bf Appendices}

\section{Derivations and additional method details}
\label{subsec:connection_EBMF_EBMN} 

The cEBMF algorithm (Sec.~\ref{sec:alg_sketch}) is a {\em variational
empirical Bayes} algorithm
%
\citep{blei_latent_2003, saul-1996, ghahramani-2000,
  van_de_wiel_learning_2019} that is formulated as solving the
following optimization problem:
\begin{equation}
\argmax_{q, \, g, \, {\bm\tau}} 
\mathrm{ELBO}(q, g, {\bm\tau}),
\end{equation}
where $g$ is shorthand for the priors $g_1^{(\ell)}, \ldots,
g_K^{(\ell)}, g_1^{(f)}, \ldots, g_K^{(f)}$, $q$ is a distribution on
$({\bf L}, {\bf F})$, and $\mathrm{ELBO}(q, g, {\bm\tau})$ is the
``Evidence Lower BOund'' (ELBO) \cite{blei_variational_2017}, a lower
bound to the ``evidence'', $\log p({\bf Z} \mid g, {\bm\tau})$:
\begin{equation}
\mathrm{ELBO}(q, g, {\bm\tau}) \colonequals 
\mathbb{E}_q[\log p({\bf Z} \mid {\bf L}, {\bf F}, {\bm\tau})] 
+ \mathbb{E}_q\bigg[\log\bigg\{\frac{p({\bf L}, {\bf F} \mid {\bf X}, {\bf Y})}
                                   {q({\bf L}, {\bf F})}\bigg\} \bigg].
\label{eq:elbo}
\end{equation}
See \citep{wang_empirical_2021, kim_flexible_2022,
  morgante_flexible_2023} for other variational empirical Bayes
algorithms derived in a similar way.

To achieve tractable update expressions for the model parameters, we
approximate the posterior $q({\bf L}, {\bf F})$ so that it factorizes
over all elements of ${\bf L}$ and all elements of ${\bf F}$
(sometimes called a ``mean field'' approximation):
\begin{equation}
\begin{aligned}
q({\bf L}, {\bf F}) &= q^{\ell}({\bf L}) \, q^f({\bf F}) \\
q^{\ell}({\bf L}) &= \prod_{i=1}^n \prod_{k=1}^K q_{ik}^{\ell}(\ell_{ik}) \\
q^f({\bf F}) &= \prod_{j=1}^p \prod_{k=1}^K q_{jk}^f(f_{jk}).
\end{aligned}
\label{eq:mean-field-approx}
\end{equation}
With this factorization (or conditional independence) qconstraint on
$q$, the right-hand part of the ELBO can be immediately decomposed
into a sum of expectations over the individual elements of ${\bf L}$
and ${\bf F}$, so we have
\begin{align}
\mathrm{ELBO}(q, g, {\bm\tau}) &= 
\mathbb{E}_q[\log(p(\bZ \mid {\bf L}, {\bf F}, {\bm\tau})] 
\nonumber \\
& \quad + \sum_{i=1}^n \sum_{k=1}^K \mathbb{E}_q \bigg[  
\log\bigg\{\frac{g_k^{(\ell)}(l_{ik}; {\bm x}_i)}
                {q_{ik}^{\ell}(\ell_{ik})}\bigg\}\bigg] 
\nonumber \\
& \quad + \sum_{j=1}^p \sum_{k=1}^K \mathbb{E}_q \bigg[ 
\log\bigg\{\frac{g_k^{(f)}(f_{jk}; {\bm y}_j)}{q_{jk}^f(f_{jk})}
\bigg\}\bigg],
\label{eq:ELBO_K_factor} 
\end{align}
where $g_k^{(\ell)}(\ell; {\bm x}_i)$ denotes the density of
$g_k^{(\ell)}({\bm x}_i)$ at $\ell$, and $g_k^{(f)}(f; {\bm y}_j)$
denotes the density of $g_k^{(f)}({\bm x}_i)$ at $f$.

\subsection{Updating the factors}
\label{sec:factors_update}

The following proposition formally connects the updates of the
individual factors $k = 1, \ldots, K$ (Step 2--4 of the algorithm in
Sec.~\ref{sec:alg_sketch}) to learning a covariate-moderated EBNM
model (Sec.~\ref{sec:ebnm}).

\begin{proposition}
\label{prop_connection}
\rm 
Let ${\bm\ell}_k = (\ell_{1k}, \ldots, \ell_{nk})^T$ denote the
$k$th column of ${\bf L}$, let ${\bm f}_k = (f_{1k}, \ldots, f_{pk})^T$
denote the $k$th column of ${\bf F}$, 
let $\bar{\bm\ell}_k = \mathbb{E}_q[{\bm\ell}_k]$,
$\bar{\bm f}_k = \mathbb{E}_q[{\bm f}_k]$,
$\bar{\bm\ell}_k^2 = \mathbb{E}_q[{\bm\ell}_k^2]$ and
$\bar{\bm f}_k^2 = \mathbb{E}_q[{\bm f}_k^2]$,
and we further define
\begin{align}
q_k^{\ell}({\bm\ell}_k) \colonequals&\; 
\prod_{i=1}^n q_{ik}^{\ell}(\ell_{ik}) \\
q_k^f({\bm f}_k) \colonequals&\; \prod_{j=1}^p q_{jk}^{f}(f_{jk}).
\end{align}
Let $\bar{\bf R}^k$ denote the $n \times p$ matrix of expected
residuals (with elements $\bar{r}_{ij}^k$) that ignores the
contribution of the $k$th factor,
\begin{equation}
\bar{r}_{ij}^{k} \colonequals 
z_{ij} - \sum_{k' \neq k} \bar{{l}}_{ik'} \bar{{f}}_{jk'}.
\end{equation}
Also define $\hat{\bm\ell}(\bZ, {\bm t}, {\bm w}, {\bm\tau})$,
$\hat{\bm f}(\bZ, {\bm t}, {\bm w}, {\bm\tau})$, ${\bm
  s}^{\ell}({\bm w}, {\bm\tau})$ and ${\bm s}^f({\bm w}, {\bm\tau})$ as
vector-valued functions in which the individual vector elements given
by
\begin{align}
\hat{\ell}_i(\bZ, {\bm t}, {\bm w}, {\bm\tau})
&= \frac{\sum_{j=1}^p \tau_{ij} z_{ij} t_j}
        {[s_i^{\ell}({\bm w}, {\bm \tau})]^2} 
\label{eq:l-hat} \\
\hat{f}_j(\bZ, {\bm t}, {\bm w}, {\bm\tau})
&= \frac{\sum_{i=1}^n \tau_{ij} z_{ij} t_i}
        {[s_j^f({\bm w}, {\bm\tau})]^2} 
\label{eq:f-hat} \\
s_i^{\ell}({\bm w}, {\bm \tau}) 
&= \textstyle (\sum_{j=1}^p \tau_{ij} w_j)^{-1/2} 
\label{eq:se-l} \\
s_j^f({\bm w}, {\bm\tau})
&= \textstyle (\sum_{i=1}^n \tau_{ij} w_i)^{-1/2}.
\label{eq:se-f}
\end{align}
Then using the definition of the ELBO in \eqref{eq:elbo} and the cEBNM
mapping defined in \eqref{eq:cEBNM_mapping}, we have that
\begin{align}
\label{eq:max_ql}
\textstyle \argmax_{q_k^{\ell}, \, g_k^{(\ell)} \,\in\, \mathcal{G}_{\ell,k}}
\mathrm{ELBO}(q, g, {\bm\tau}) &= 
\mathrm{cEBNM}(\hat{\bm\ell}(\bar{\bf R}^k, \bar{\bm f}_k, 
\bar{\bm f}_k^2, {\bm\tau}), 
{\bm s}_{\bm\ell}(\bar{\bm f}_k^2, {\bm\tau}), \bX, \mathcal{G}_{\ell,k}) \\
\label{eq:max_qf}
\textstyle \argmax_{q_k^f, \, g_k^{(f)} \,\in\, \mathcal{G}_{f,k}} 
\mathrm{ELBO}(q, g, {\bm\tau}) &= 
\mathrm{cEBNM}(\hat{\bm f}(\bar{\bf R}^k, \bar{\bm\ell}_k, 
\bar{\bm\ell}_k^2, {\bm\tau}), {\bm s}_{\bm f}(\bar{\bm\ell}_k^2, 
\bm \tau), \bY, \mathcal{G}_{f,k}).
\end{align}
Note that this identity requires a slight change to the definition of
the cEBNM mapping \eqref{eq:cEBNM_mapping} as returning the priors
$g({\bm d}_i, \hat{\bm\theta})$ at $\hat{\bm\theta}$ rather than the
parameter estimates $\hat{\bm\theta}$ themselves.
\end{proposition}

\begin{proof}  
Starting from \eqref{eq:ELBO_K_factor}, 
we expand on the parts of ELBO that involve $q_k^{\ell}$ or
$g_k^{(\ell)}$ or both:
\begin{equation}
\mathrm{ELBO}(q, g, {\bm\tau}) 
= -\frac{1}{2}\sum_{i=1}^n 
\mathbb{E}_q[a_{ik} {l}_{ik}^2 - 2 b_{ik}{l}_{ik}]
+ \sum_{i=1}^n \mathbb{E}_q \bigg[
\log\bigg\{\frac{g_k^{(\ell)}(\ell_{ik}; {\bm x}_i)}
                {q_{ik}^{(\ell)}(\ell_{ik})}\bigg\}\bigg]
+ \mbox{const},
\label{eq:elbo_expanded}
\end{equation}
where ``const'' is a placeholder for the terms in the ELBO that do not
depend on the $k$th factor, and we define
\begin{align}
a_{ik} \colonequals&
\sum_{j=1}^p \tau_{ij} (\bar{f}_{jk})^2 \\
b_{ik} \colonequals& 
\sum_{j=1}^p \tau_{ij} \bar{r}_{ij}^k \bar{f}_{jk}.
\end{align}
The identity \eqref{eq:max_ql} then follows from Lemma
\ref{lemma:cebnm} (given below). The other identity \eqref{eq:max_qf}
is proved similarly.
\end{proof}

\subsection{Updating the residual variances}
\label{sec:resid_var_update}

Focussing on the part of the ELBO depends on ${\bm\tau}$, we have
\begin{equation}
\mathrm{ELBO}(q, g, {\bm\tau}) = 
\frac{1}{2}
\sum_{i=1}^n \sum_{j=1}^p (\log\tau_{ij}  
- \tau_{ij} \bar{r}_{ij}^2)  + \mbox{const},
\label{eq:elbo-tau}
\end{equation}
in which ``const'' is a placeholder for the other terms in the ELBO
that do not involve ${\bm\tau}$, and $\bar{r}_{ij}^2$ is the expected
squared difference between the observation $z_{ij}$ and the value
predicted by the matrix factorization:
\begin{equation}
\bar{r}_{ij}^2 \colonequals \mathbb{E}_q[(z_{ij} - \hat{z}_{ij})^2],
\end{equation}
where
\begin{equation}
\hat{z}_{ij} = \sum_{k=1}^K l_{ik} f_{jk}.
\end{equation}

If one makes the modeling assumption that all the residual variances
are the same, i.e., $\tau_{ij} = \tau$, then from
\eqref{eq:elbo-tau} the update for $\tau$ works out to
\begin{equation}
\tau = \frac{n \times p}{\sum_{i=1}^n \sum_{j=1}^p \bar{r}_{ij}^2}.
\label{eq:tau-update-all-equal}
\end{equation}
If instead one makes the weaker modeling assumption that the residual
variances are the same in each column, i.e., $\tau_{ij} = \tau_j, j =
1, \ldots, p$, then the updates work out to
\begin{equation}
\tau_j = \frac{n}{\sum_{i=1}^n \bar{r}_{ij}^2}.
\label{eq:tau-update-rows-equal}
\end{equation}
Similarly, for row-specific residual variances the updates are
\begin{equation}
\tau_i = \frac{p}{\sum_{j=1}^p \bar{r}_{ij}^2}.
\label{eq:tau-update-cols-equal}
\end{equation}
For all these expressions, the squared differences $\bar{r}_{ij}^2$
are easily computed given the conditional independence assumptions of
the fully-factorized approximation \eqref{eq:mean-field-approx}:
\begin{equation}
\bar{r}_{ij}^2 = 
\bigg(z_{ij} - \sum_{k=1}^K \bar{l}_{ik} \bar{f}_{jk}\bigg)^2 
+ \sum_{k=1}^K (\bar{l}_{ik}^2) (\bar{f}_{jk}^2)
- \sum_{k=1}^K (\bar{l}_{ik} \bar{f}_{jk})^2,
\label{eq:expected_res}
\end{equation}
in which we have defined $\bar{l}_{ik} \colonequals
\mathbb{E}_q[l_{ik}]$, $\bar{f}_{jk} \colonequals
\mathbb{E}_q[f_{jk}]$, $\bar{l}_{ik}^2 \colonequals
\mathbb{E}_q[l_{ik}^2]$ and $\bar{f}_{jk}^2 \colonequals
\mathbb{E}_q[f_{jk}^2]$.

\subsection{Covariate-moderated EBNM}

To complete the proof of Proposition \ref{prop_connection}, it remains
to show that the identity \eqref{eq:max_ql} is satisfied at the
objective function given by \eqref{eq:elbo_expanded}. (And similarly
for the identity \eqref{eq:max_qf}.) This connection is made in the
following lemma.

\begin{lemma}
\label{lemma:cebnm}
\rm 
Consider the cEBNM mapping defined in \eqref{eq:cEBNM_mapping}. An
equivalent definition of this mapping is
\begin{equation}
\textstyle 
(\hat{\bm\theta}, \hat{q}) = 
\argmax_{{\bm\theta}, \, q} 
F({\bm\theta}, q; \hat{\bm\beta}, {\bm s}, {\bf D}),
\end{equation}
where 
\begin{equation}
\label{eq:objective-lemma}
F({\bm\theta}, q; \hat{\bm\beta}, {\bm s}, {\bf D}) 
= -\frac{1}{2} \sum_{i=1}^n  \mathbb{E}_q[
a_i \beta_i^2 - 2b_i\beta_i] + 
\sum_{i=1}^n \mathbb{E}_q \bigg[\log\bigg\{
\frac{g(\beta_i; {\bm d}_i, {\bm\theta})}{q_i(\beta_i)}
\bigg\}\bigg],
 \end{equation}
and $g(\beta; {\bm d}_i, {\bm\theta})$ denotes the density of $g({\bm
  d}_i, {\bm\theta})$ at $\beta$, and we further define
\begin{align}
q({\bm\beta}) &= \textstyle \prod_{i=1}^n q_i(\beta_i) \\
a_i &= 1/s_i^2 \\
b_i &= \hat{\beta}_i/s_i^2.
\end{align}
\end{lemma}
 
\begin{proof}
We begin with the ELBO for the cEBNM model (\ref{eqn:nm_problem},
\ref{eq:cebnm_prior}):
%
%
%
\begin{align}
\mathrm{ELBO}({\bm\theta}, q; \hat{\bm\beta}, {\bm s}, {\bf D}) = 
\log\mathcal{L}({\bm\theta}) 
- D_{\mathrm{KL}}(q \,\|\, p_{\mathsf{post}}).
\end{align}
where $D_{\mathrm{KL}}(q \,\|\, p)$ denotes the Kullback-Leibler (K-L)
divergence from a distribution $p$ to a distribution $q$
\citep{kullback1951information}, $\mathcal{L}({\bm\theta})$ is the
marginal likelihood defined in \eqref{eqn:lik}, and
$p_{\mathsf{post}}({\bm\beta})$ is the (exact) posterior distribution,
$p_{\mathsf{post}}({\bm\beta}) \colonequals \prod_{i=1}^n p(\beta_i
\mid \hat{\beta}_i, s_i, \hat{\bm\theta}, {\bf D})$ (see
eq.~\ref{eqn:ebnm_post}). Since $D_{\mathrm{KL}}(q \,\|\, p)$ is
always zero or greater, and is exactly zero when $p = q$, we have that
$\argmax_q \mathrm{ELBO}({\bm\theta}, g; \hat{\bm\beta}, {\bm s}, {\bf
  D}) = p_{\mathsf{post}}$ and $\max_q \mathrm{ELBO}({\bm\theta}, g;
\hat{\bm\beta}, {\bm s}, {\bf D}) = \log
\mathcal{L}({\bm\theta})$. Next, a basic identity of the ELBO (see for
example Appendix B of \citep{wang_simple_2020}) is that the ELBO can
be rewritten as
\begin{equation}
mathrm{ELBO}({\bm\theta}, q; \hat{\bm\beta}, {\bm s}, {\bf D}) =
\mathbb{E}_q[\log p(\hat{\bm\beta} \mid {\bm\beta}, {\bm s})] +
\sum_{i=1}^n \mathbb{E}_q\bigg[\log \bigg\{
\frac{g(\beta_i; {\bm d}_i, {\bm\theta})}{q_i(\beta_i)}
\bigg\}\bigg].
\label{eq:objective-lemma-almost-there}
\end{equation}
To complete the proof, we expand terms in the log-likelihood in 
\eqref{eq:objective-lemma-almost-there}:
\begin{equation}
\log p(\hat{\bm\beta} \mid {\bm\beta}, {\bm s})
= -\frac{1}{2}\sum_{i=1}^n \frac{(\beta_i - \hat{\beta}_i)^2 }{s_i^2}
+ \mbox{const},
\end{equation}
where ``const'' is a placeholder for terms that do not involve $q$ (or
$g$). Plugging this identity into
\eqref{eq:objective-lemma-almost-there}, and with a bit of additional
algebraic manipulation, we recover \eqref{eq:objective-lemma}.
\end{proof}

\subsection{Detailed algorithms}

In summary, the cEBMF algorithm is a block co-ordinate ascent
algorithm \citep{wright2015coordinate} for finding a local maximum of
the ELBO \eqref{eq:elbo}, in which the ``blocks''---i.e., the subsets
of parameters to be updated---are the individual factors $k = 1,
\ldots, K$ (Sec.~\ref{sec:factors_update}) and the residual variances
${\bm\tau}$ (Sec.~\ref{sec:resid_var_update}). This co-ordinate ascent
algorithm is described in Algorithm \ref{alg:Backfitting_update}, and
the single-factor update is described in Algorithm
\ref{alg:single_update}. (And it is described informally in
Sec.~\ref{sec:alg_sketch}.) In practice, we run Algortithm
\ref{alg:Backfitting_update} until the increase in the ELBO across two
successive iterations is smaller than some specified tolerance, or
until we have reached an upper bound on the number of iterations.

\begin{algorithm}[t]
\caption{cEBMF algorithm}
\label{alg:Backfitting_update}
\begin{algorithmic}

\REQUIRE $n \times p$ data matrix, $\bZ$; covariate or ``side
information'' matrices, ${\bf X}$ $(n \times n_x)$ and ${\bf Y}$ $(p
\times n_y)$; $K$, the number of factors; 
the prior families $\mathcal{G}_{\ell,k}$ and $\mathcal{G}_{f,k}$,
$k = 1, \ldots, K$; and initial estimates of the
first and second moments of ${\bf L}$ $(n \times K)$, ${\bf F}$ $(p
\times K)$, which are denoted by $\bar{\bf L}$, $\bar{\bf F}$,
$\bar{\bf L}^2$, $\bar{\bf F}^2$.

\STATE Compute the expected residuals, 
$\bar{\bf R} = {\bf Z} - \bar{\bf L} \bar{\bf F}^T$.

\REPEAT 

\STATE Update the residual variances ${\bm\tau}$ using
\eqref{eq:tau-update-all-equal}, \eqref{eq:tau-update-rows-equal} or
\eqref{eq:tau-update-cols-equal}.

\FOR{$k = 1, \ldots, K$}

\STATE Remove the effect of the $k$th factor from the expected 
  residuals, $\bar{\bf R}^k = \bar{\bf R} + \bar{\bm\ell}_k \bar{\bm f}_k^T$.

\STATE Perform a single-factor update for factor $k$ 
(Algorithm \ref{alg:single_update}).

\STATE Update the expected residuals, 
$\bar{\bf R} = \bar{\bf R}^k - \bar{\bm\ell}_k \bar{\bm f}_k^T$.

\ENDFOR

\UNTIL some convergence criterion is met

\RETURN $\bar{\bf L}, \bar{\bf F}, {\bm\tau},
g_1^{(\ell)}, \ldots, g_K^{(\ell)}, g_1^{(f)}, \ldots, g_K^{(f)}$

\end{algorithmic}
\end{algorithm}

\algsetup{linenodelimiter=.}
\begin{algorithm}[t]
\caption{cEBMF single-factor update}
\label{alg:single_update}
\begin{algorithmic}[1]

\REQUIRE covariate or ``side information'' matrices, ${\bf X}$ $(n
\times n_x)$ and ${\bf Y}$ $(p \times n_y)$; $k \in \{1, \ldots, K\}$,
the dimension to update; the prior families $\mathcal{G}_{\ell,k}$ and
$\mathcal{G}_{f,k}$; an implementation of $\mathrm{cEBNM}(\hat{\bm
  \beta}, {\bm s}, {\bf D}, \mathcal{G}) \rightarrow (\hat{\bm\theta},
\hat{q})$ (eq.~\ref{eq:cEBNM_mapping}) for prior families $\mathcal{G}
= \mathcal{G}_{\ell,k}$ and $\mathcal{G} = \mathcal{G}_{f,k}$; the
expected residuals, $\bar{\bf R}^k$; estimates of the second moments,
$\bar{\bf L }^2$ $(n \times K)$, $\bar{\bf F}^2$ $(p \times K)$; and
the residuals variances, ${\bm\tau}$.

\STATE $\hat{\bm\beta} \leftarrow 
\hat{\bm\ell}(\bar{\bf R}^k, \bar{\bm f}_k, \bar{\bm f}_k^2, {\bm\tau})$

\STATE ${\bm s} \leftarrow {\bm s}_{\ell}(\bar{\bm f}_k^2, {\bm\tau})$

\STATE $(g_k^{(\ell)}, q_k^{\ell}) \leftarrow
\mathrm{cEBNM}(\hat{\bm\beta}, {\bm s}, {\bf X}, \mathcal{G}_{\ell,k})$

\STATE Compute posterior moments
$\bar{\ell}_{ik} \colonequals E_q[\ell_{ik}]$ and 
$\bar{\ell}_{ik}^2 \colonequals E_q[\ell_{ik}^2]$,
$i = 1, \ldots, n$.

\STATE $\hat{\bm\beta} \leftarrow \hat{\bm f}(\bar{\bf R}^k,
\bar{\bm\ell}_k, \bar{\bm\ell}_k^2, {\bm\tau})$ 

\STATE ${\bm s} \leftarrow {\bm s}_f(\bar{\bm\ell}_k^2, {\bm\tau})$

\STATE $(g_k^{(f)}, q_k^{f}) \leftarrow
\mathrm{cEBNM}(\hat{\bm\beta}, {\bm s}, {\bf Y}, \mathcal{G}_{f,k})$

\STATE Compute posterior moments
$\bar{f}_{jk} \colonequals E_q[f_{jk}]$ and 
$\bar{f}_{jk}^2 \colonequals E_q[f_{jk}^2]$, 
$j = 1, \ldots, p$.

\RETURN $\bar{\bm\ell}_k, \bar{\bm\ell}_k^2, 
\bar{\bm f}_k, \bar{\bm f}_k^2, g_k^{(\ell)}, g_k^{(f)}$

\end{algorithmic}
\end{algorithm}

Two features of the empirical Bayes approach to matrix factorization
discussed in \citep{wang_empirical_2021} are worth highlighting here.
First, there is a simple stepwise procedure for obtaining good initial
estimates of ${\bf L}$ and ${\bf F}$ by introducing the factors
sequentially. This was called a ``greedy initialization'' in
\cite{wang_empirical_2021}. Second, instead of fixing the number of
factors, the EBMF approach can also select $K$ automatically by
adapting the priors $g_k^{(\ell)}, g_k^{(f)}$ separately for each
factor $k$. The idea is that factors that are not useful for
explaining variation in the data should produce priors that are
concentrated near zero (this feature of course requires that the
chosen prior families $\mathcal{G}_{\ell,k}, \mathcal{G}_{f,k}$
include distributions that are concentrated near zero). Therefore, $K$
can initially be set to a large value, and the cEBMF algorithm will
automatically determine an appropriate number of factors by
``shrinking'' the unneeded factors.


\subsubsection{Computational complexity}

Since cEBMF is a modeling and algorithmic framework, and not a
specific method or algorithm, we cannot give the exact computational
complexity of Algorithm \ref{alg:Backfitting_update}. However, we can
provide some rules of thumb.  Steps 1, 2, 5 and 6 in Algorithm
\ref{alg:single_update} (also, Steps 1 and 2 in
Sec. \ref{sec:alg_sketch}) involve preparing the inputs for the cEBNM
solver \eqref{eq:cEBNM_mapping}.  Since these steps do not depend on
the prior families $\mathcal{G}_{\ell,k}, \mathcal{G}_{f,k}$, we can
give their computational complexity: when ${\bf Z}$ is a ``dense''
(non-sparse) matrix, the time complexity for updating a single factor
$k$ is $O(np)$; when ${\bf Z}$ is sparse, the complexity is $O(S)$,
where $S$ is the number of nonzero entries in ${\bf Z}$. (Note this
requires careful implementation that avoids directly storing $\bar{\bf
  R}$). Steps 3 and 7 in Algorithm \ref{alg:single_update} (or Steps 3
and 4 in Sec. \ref{sec:alg_sketch}) will depend on the details of the
cEBNM solver and the type of side information. However, when the
priors on ${\bf L}, {\bf F}$ are simple and involve low-dimensional
covariates, the other steps are expected to dominate, in which case
the complexity of Algorithm \ref{alg:Backfitting_update} is expected
to be $O(npK)$ or $O(SK)$.


\section{Details of the experiments}
\label{sec:experiment_details}

\subsection{Simulations}


We simulated data sets from different cEBMF models. In all cases, the
data were generated with homoskedastic noise, $\tau_{ij} = \tau$.

\paragraph*{Sparsity-driven covariate simulations.} 

This simulation was intended to illustrate the behaviour of cEBMF when
provided with simple row and column-covariates that inform only the
sparsity of ${\bf L}$ and ${\bf F}$ (and not the magnitudes of their
elements). The side information was stored in $\mbox{1,000} \times 10$
and $200 \times 10$ matrices ${\bf X}$ and ${\bf Y}$, and the
$\mbox{1,000} \times 200$ matrix ${\bf Z}$ was simulated using a
simple cEBMF model with $K = 2$
%
%
and with spike-and-slab priors chosen to ensure that 90\% of the
elements of ${\bf L}{\bf F}^T$ were zero. Specifically, we used the
following priors:
%
%
\begin{equation}
\begin{aligned}
\ell_{ik} \sim&\; 
\pi_{ik} \delta_0 + (1 - \pi_{ik}) N(0,1) \\
f_{jk} \sim&\; \alpha_{jk} \delta_0 + (1 - \alpha_{jk}) N(0,1)\\
\pi_{ik} \colonequals&\; 
\phi({\bm\theta}_k^T\bx_i) \\
\alpha_{jk} \colonequals&\;
\phi({\bm\omega}_k^T\by_j),
\end{aligned}
\end{equation}
%
%
in which ${\bm\theta}_k$ and ${\bm\omega}_k$ were chosen to achieve
90\% zeros in ${\bf L}{\bf F}^T$.
 
%
%
%
%

\paragraph*{Uninformative covariate simulations.} 

To verify that cEBMF was robust to situations in which the side
information was not helpful, we considered an ``uninformative
covariate'' setting in which the covariates were just noise. The data
sets were simulated in the same way as in the sparsity-driven covariate
simulations except that 
%
%
the true factors were simulated as $\ell_{ik} \sim \pi \delta_0 +
(1-\pi) N(0,1)$, $f_{jk} \sim \alpha \delta_0 + (1-\alpha) N(0,1)$,
with $\pi, \alpha$ chosen to achieve a target sparsity of 90\% zeros
in ${\bf L}{\bf F}^T$.



\paragraph*{Tiled-clustering simulations.} 

In this setting, we simulated rank-3 matrix factorizations in which
${\bf L}$---but not ${\bf F}$---depended on the 2-d locations of the
data points. (One of these simulations is shown in
Fig.~\ref{fig:toy_example}.) This was accomplished as follows. First,
we generated a periodic tiling of $[0,1] \times [0,1]$, randomly
labeling each tile 1, 2 or 3. For each data point $i = 1, \ldots,
\mbox{1,000}$, we sampled its 2-d location uniformly from $[0,1]
\times [0,1]$, then we set $\ell_{ik} = 1$ if the data point fell in
the tile with label $k$, otherwise $\ell_{ik} = 0$. The $200 \times 3$
matrix ${\bf F}$ was simulated from a scale mixture of zero-centered
normals that did not depend on tile membership.
%
%
%
%

\paragraph*{Shifted tiled-clustering simulations.}

To assess robustness to model misspecification, we simulated data using
a prior that could not be recovered by the prior family we used in
cEBMF. These simulations were the same as the tiled-clustering
simulations except that we generated the $i$th row of ${\bf L}$ as
follows: $(1, 2, 3)$ if data point was $i$ in the tile with label 1;
$(3, 1, 2)$ if data point was in the tile with label 2, and $(2, 3,
1)$ if data point was in the tile with label 3.

\subsection{Additional details on the methods compared}


We first describe how the methods were run on the simulated data
sets. Modifications to the methods for the MovieLens and spatial
transcriptomics data are given in the main text, with additional
technical details below.
%
%
For all methods, when possible to do so, we set the rank, $K$, to
match the rank of the simulated matrix factorization.

For EBMF and cEBMF, we assumed homoskedastic noise ($\tau_{ij} =
\tau$), and the prior families were chosen to align with how the data
were simulated (except for the shifted-tiled clustering simulations,
which were intended to illustrate the methods' behaviour when the
priors were misspecified). For EBMF, the prior families were all
elaborations of the ``spike and slab'' priors
(Sec.~\ref{sec:cebmf-spike-and-slab}). For cEBMF, the priors were of
the same form as EBMF in which the mixture weights were parameterized
using either a multinomial regression (i.e., a single-layer neural
network with a softmax link function) or a multilayer perceptron. 

In the sparsity-driven covariate and uninformative covariate
simulations, the priors for ${\bf L}$ and ${\bf F}$ in EBMF were all
scale mixtures of normals with a fixed grid of scales
\citep{stephens_false_2017, kim_fast_2019}. For cEBMF, we used priors
of the same form, except that side information was incorporated into
the prior mixture weights as follows using priors of the following form:
\begin{equation}
g({\bm d}_i, {\bm\theta}) =
\pi_0({\bm d}_i, {\bm\theta} ) \delta_0 +
\sum_{m=1}^M \pi_m ({\bm d}_i, {\bm\theta}) N(0, \sigma_m^2).
\label{eq:mixture-of-normals-prior-cebmf}
\end{equation}
The mixture weights $\pi_0, \ldots, \pi_M$ were implemented using a
standard multinomial regression model with the softmax link function.
%
%

In the tiled-clustering and shifted tiled-clustering simulations, the
true ${\bf L}$ was always non-negative. Therefore, we chose the prior
families in EBMF and cEBMF to produce {\em semi-non-negative matrix
  factorizations} \citep{ding_convex_2010} with a non-negative ${\bf
  L}$. Specifically, we assigned mixture-of-exponentials priors to
${\bf L}$, similar to the scale-mixture-of-normals priors, except with
support for non-negative numbers only \cite{ebnm}. And we assigned the
scale-mixture-of-normal priors, same as above, to ${\bf F}$.  In
cEBMF, side information was incorporated into the mixture weights in
the prior in a manner similar to above:
\begin{equation}
g({\bm d}_i, {\bm\theta}) =
\pi_0({\bm d}_i, {\bm\theta}) \delta_0 +
\sum_{m=1}^M \pi_m({\bm d}_i, {\bm\theta}) \exp(\lambda_m),
\label{eq:mixture-of-exponentials-prior-cebmf}
\end{equation}
in which $\exp(\lambda)$ denotes the exponential distribution with
scale parameter $\lambda$, and $\lambda_{m-1} < \lambda_m$, $m = 2,
\ldots, M$. As before, the mixture weights $\pi_0, \ldots, \pi_M$
were implemented using a standard multinomial regression model with
the softmax link function.


The deep learning methods (VAE, cVAE, NCF) were all implemented in
PyTorch. All the models were trained for 50 epochs using the Adam
optimizer with learning rate 0.001 and batch size 64. VAE had three
hidden layers (of width 128, 64 and 30) in both the encoder and
decoder (20 hidden dimensions). ReLU activations were used
throughout. We use the ELBO from \cite{kingma_auto-encoding_2014} to
train the model. We proceeded similarly for the cVAE, conditioning
both the encoder and decoder on the available covariate data ${\bf X},
{\bf Y}$. For cVAE, we used the training objective from
\cite{sohn_learning_2015}. NCF models ${\bf Z}$ using two separate
multilayer perceptrons for the row and column covariates
\cite{he_neural_2017}.
%
%
The multilayer perceptrons were implemented in a similar way to the VAE
encoders and decoders; that is, three hidden layers (of width 128, 64
and 30) with RELU activations.  The penalty parameters in PMD were
tuned via cross-validation as recommended by the authors. SSVD (R
package ``ssvd'') was run with its default values.




For the spatial transcriptomics data, we fit cEBMF and EBMF using
gene-specific residual variances, $\sigma_{ij}^2 = \sigma_j^2$. We
used mixture-of-exponential priors for ${\bf F}$, and the
parameterized mixture-of-exponential priors
\eqref{eq:mixture-of-exponentials-prior-cebmf} for ${\bf L}$ in which
the mixture weights were learned using a multilayer perceptron instead
of a multinomial regression.  The multilayer perceptions were defined
as sequential models with a dense layer with 64 units and ReLU
activations. We use two subsequent dense layers, each with 64 units,
and ReLU activations using an L2 regularization coefficient of 0.001 to
prevent overfitting.
%
%
These regularized layers were followed by a dropout layer (with a
dropout rate of 0.5). The subsequent layers were four dense layers
each with 64 units, ReLU activations and L2 regularization coefficient
of 0.001. The final layer was a dense layer with a softmax
activation. These models were trained during each single-factor update
using 300 epochs and a batch size of 1,500.

In the simulations, cEBMF was implemented in R, in which learning the
parameterized priors was performed using the Keras R interface
\citep{chollet_keras_2015} to TensorFlow
\citep{martin_abadi_tensorflow_2015}. For the MovieLens and spatial
transcriptomics data sets, we used the PyTorch-based implementation of
cEBMF which we have made available as a Python package on GitHub.



%
%


%

\clearpage

\section{Additional results}
\label{app:sup_fig}


\begin{figure}[ht]
\centering
\includegraphics[width=3.65in]{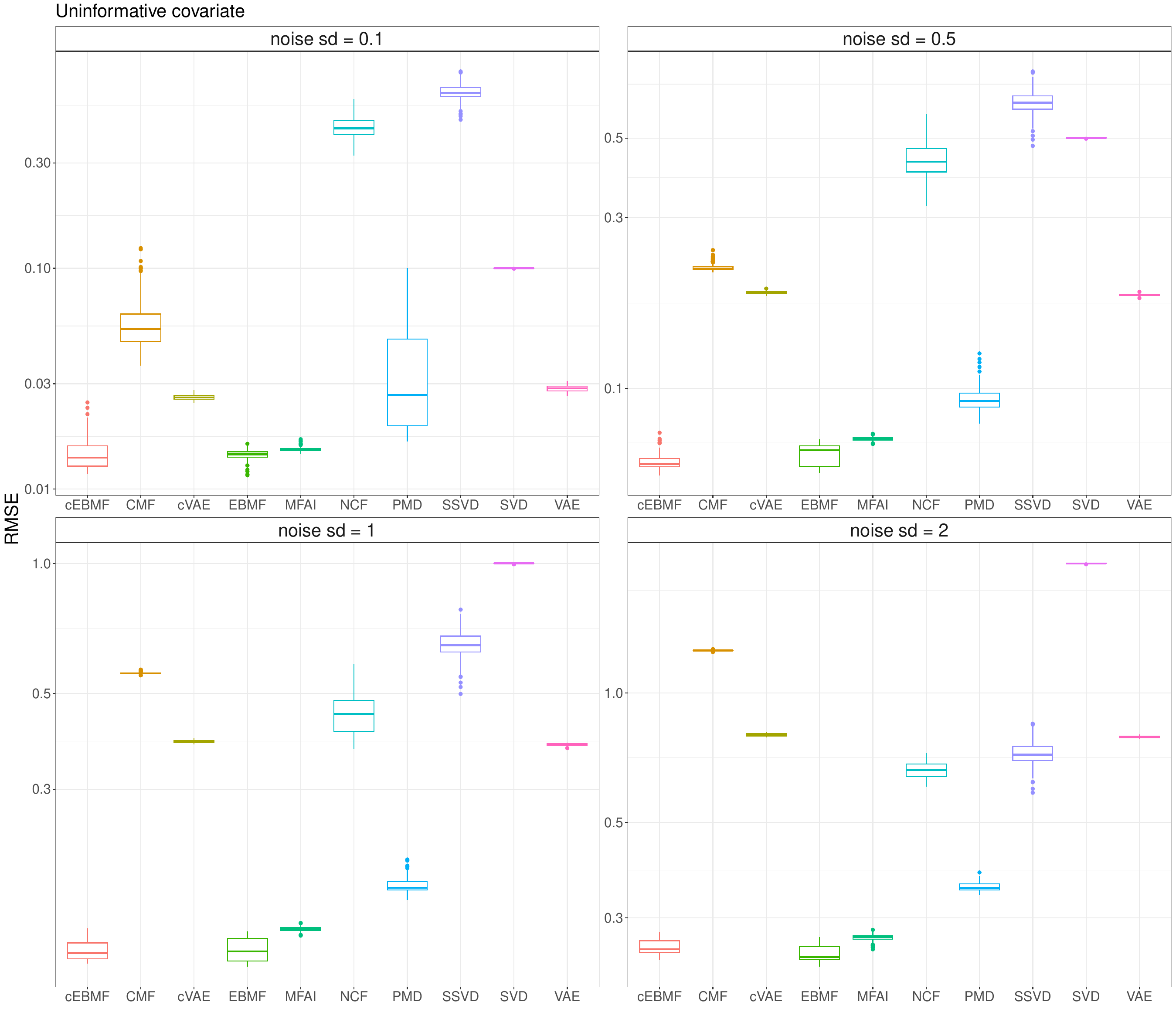}
\caption{Simulation results from the ``uninformative covariate''
  setting in which the data were simulated under different noise
  levels, $\tau$. Note that for improved visualization the RMSE is
  shown on the log-scale.}
%
%
\label{fig:uninformative-sims}
\end{figure}

\begin{figure}[ht]
\centering
\includegraphics[width=3.65in]{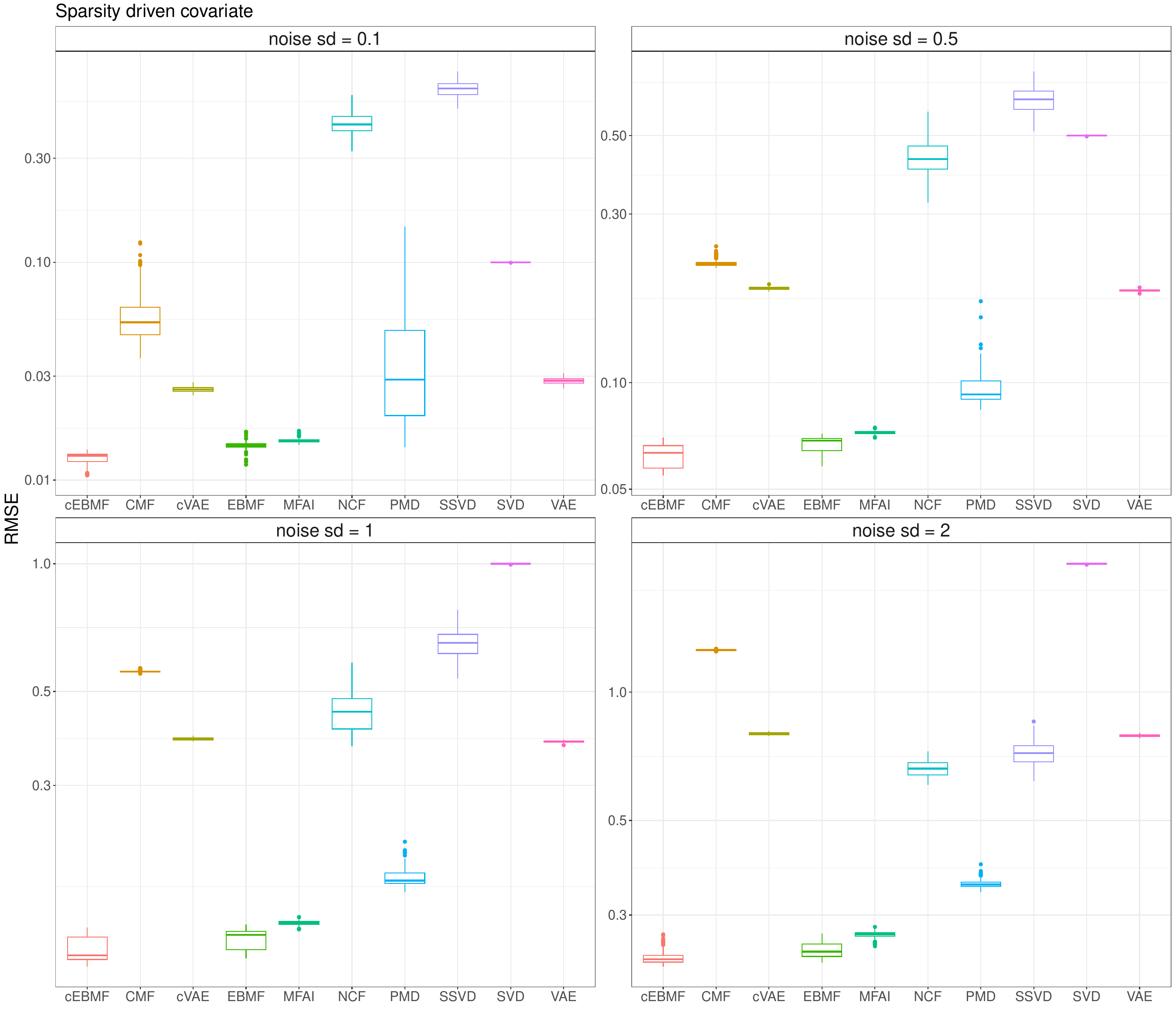}
\caption{Simulation results from the ``sparsity-driven covariate''
  setting in which the data were simulated under different noise
  levels, $\tau$. Note that for improved visualization the RMSE is
  shown on the log-scale.}
  %
  %
  %
\label{fig:sparsity-driven-sims}
\end{figure}

\begin{figure}[ht]
\centering
\includegraphics[width=3.75in]{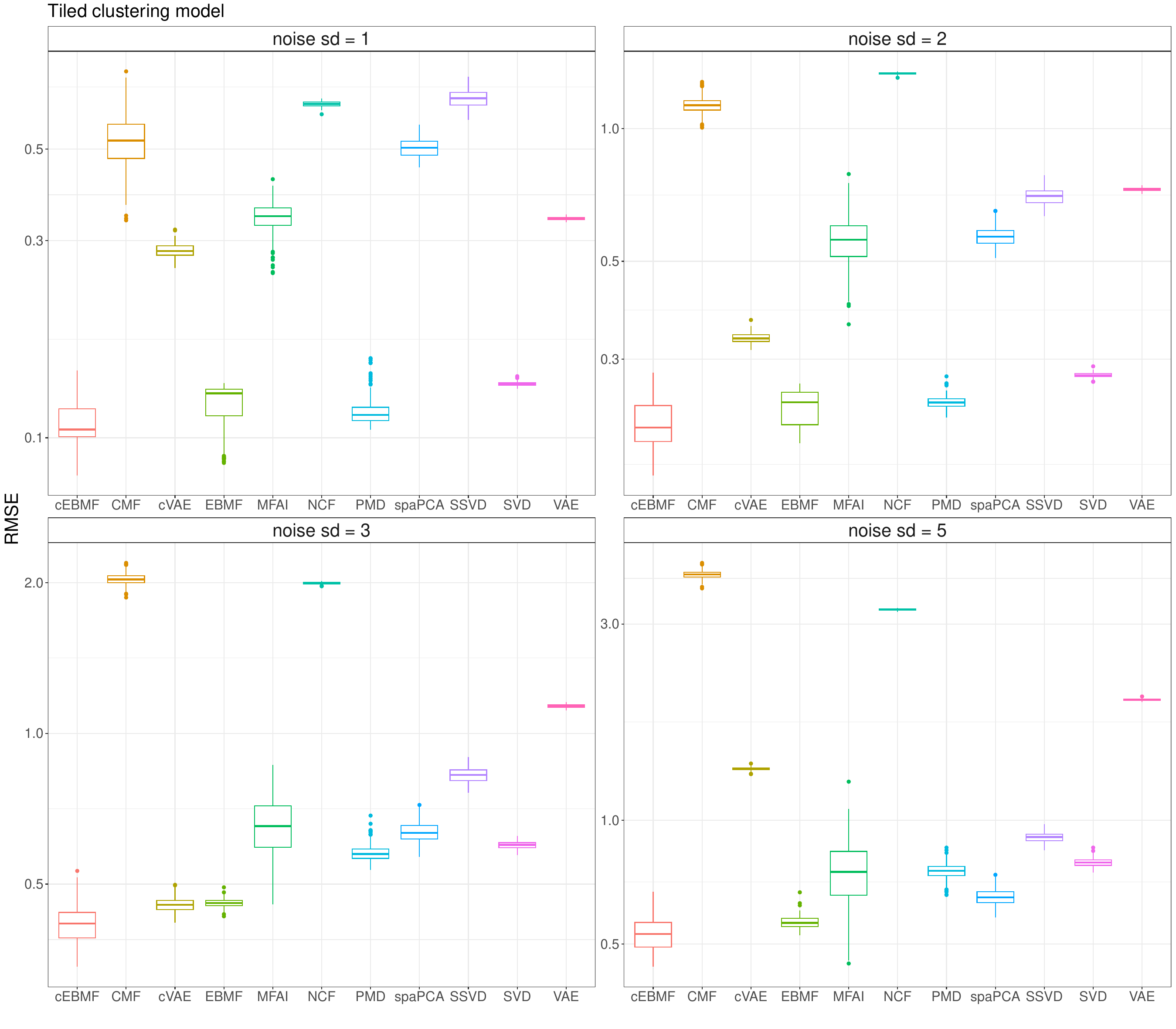}
\caption{Simulation results from the ``tiled-clustering'' setting in
  which the data were simulated under different noise levels,
  $\tau$. Note that for improved visualization the RMSE is shown on
  the log-scale. (spaPCA = Spatial PCA)}
\label{fig:tiling-sims}
\end{figure}

\begin{figure}[ht]
\centering
\includegraphics[width=3.65in]{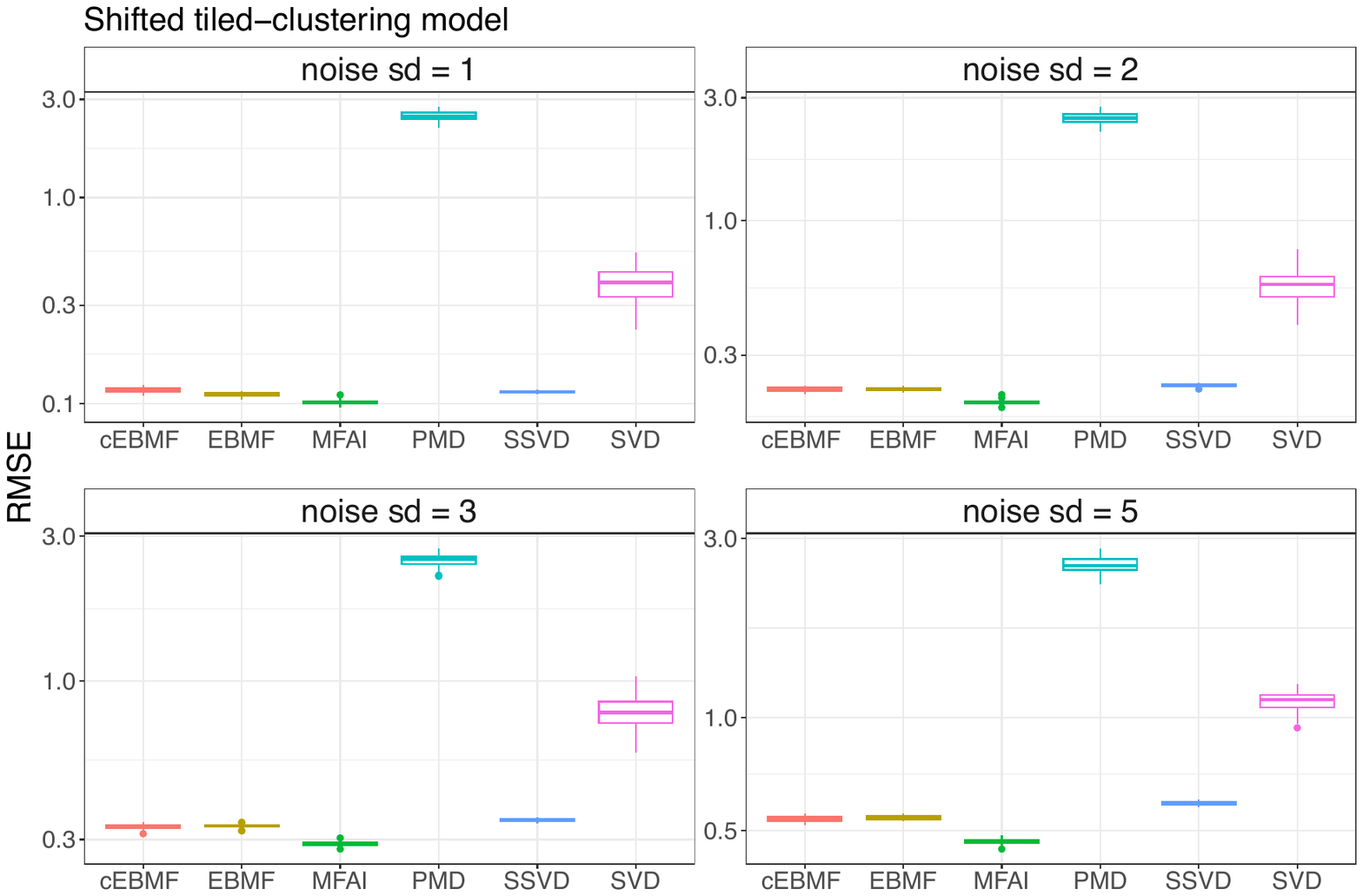}
\caption{Simulation results from the ``shifted tiled-clustering''
  setting in which the data were simulated under different noise
  levels, $\tau$. Note that for improved visualization the RMSE is
  shown on the log-scale.}
%
%
\label{fig:shifted-tiled-sims}
\end{figure}


\label{supsec:spaRNA}

\begin{figure}[t]
\centering 
\includegraphics[width=\textwidth]{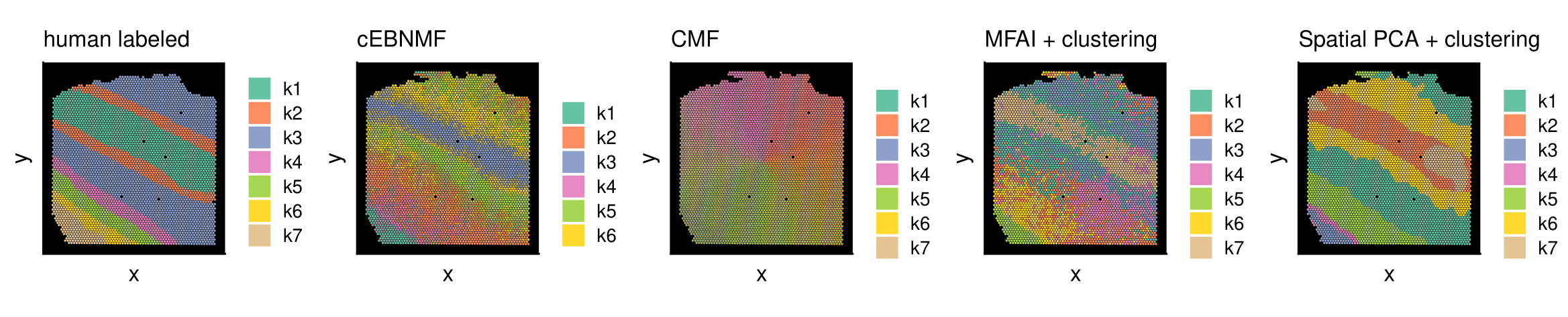}
\includegraphics[width=\textwidth]{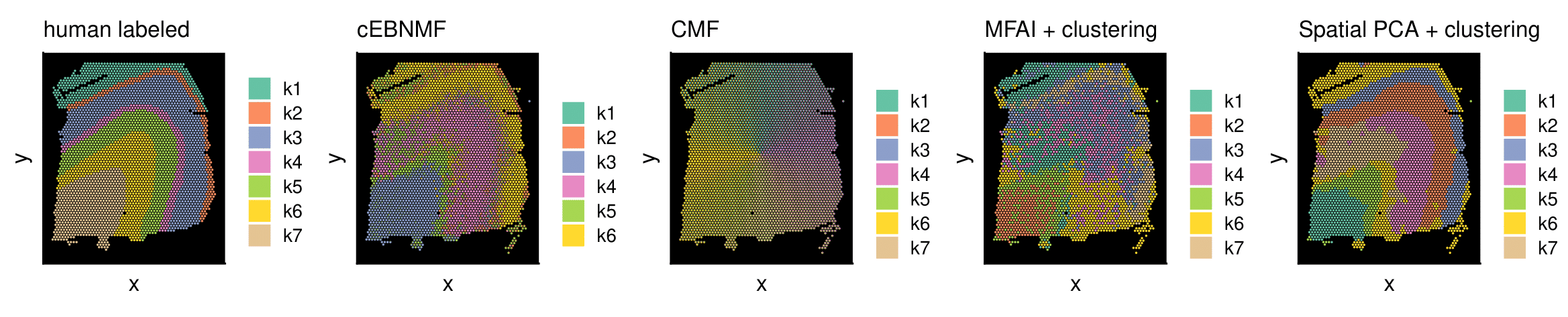}
\caption{Additional results on slides 4 (top) and 10 (bottom) of the
  DLPFC spatial transcriptomics data. See
  Fig.~\ref{fig:spatial_soft_clustering} for additional information
  about these results.}
  %
  %
  %
\label{fig:sprna_mfai_cmf}
\end{figure}

\begin{figure}[ht]
\centering
\includegraphics[width=\textwidth]{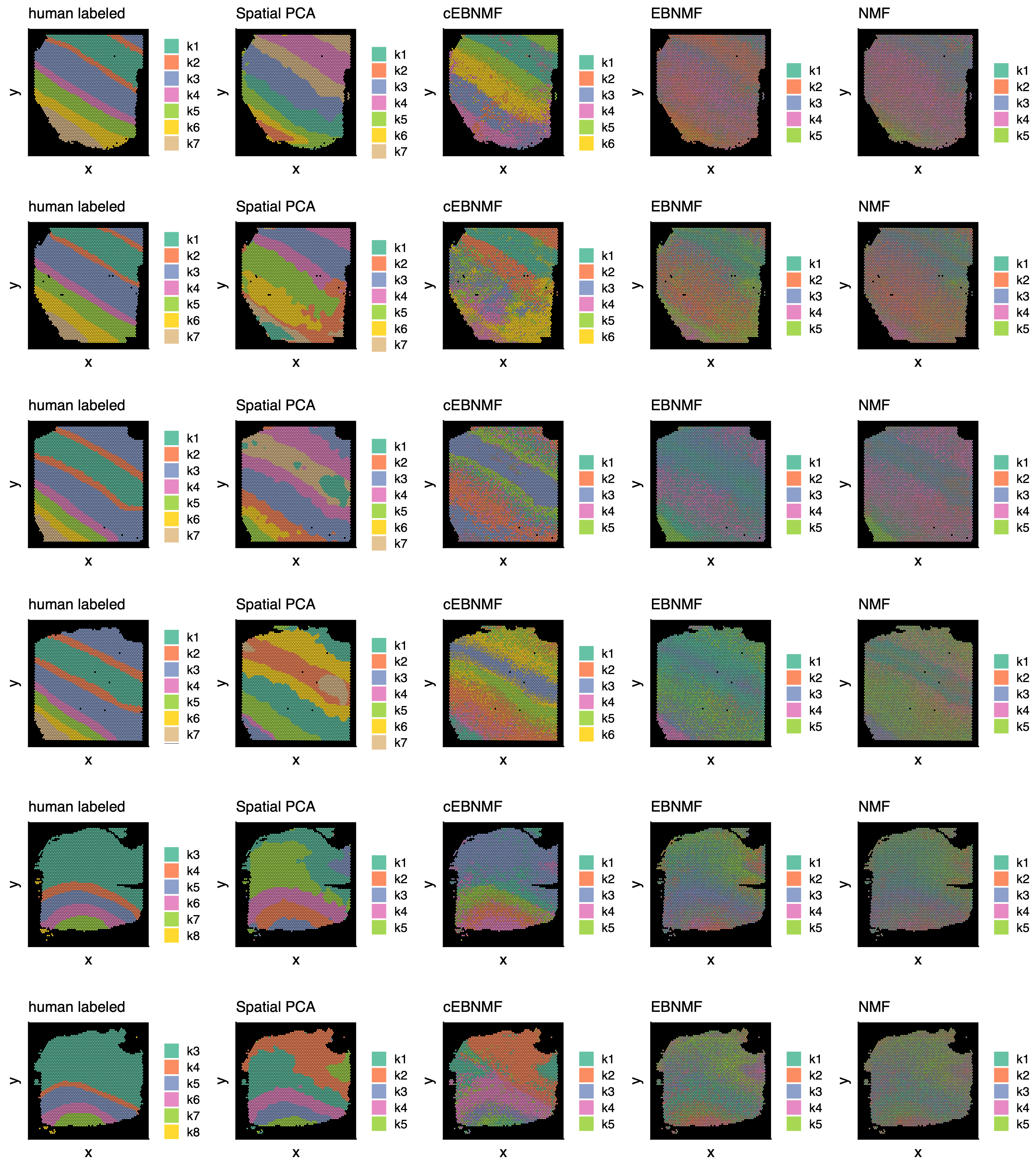}
\caption{Selected results on slices 1 (top row) through 6 (bottom row)
  of the DLPFC spatial transcriptomics data.}
\end{figure}

\begin{figure}[ht]
\centering
\includegraphics[width=\textwidth]{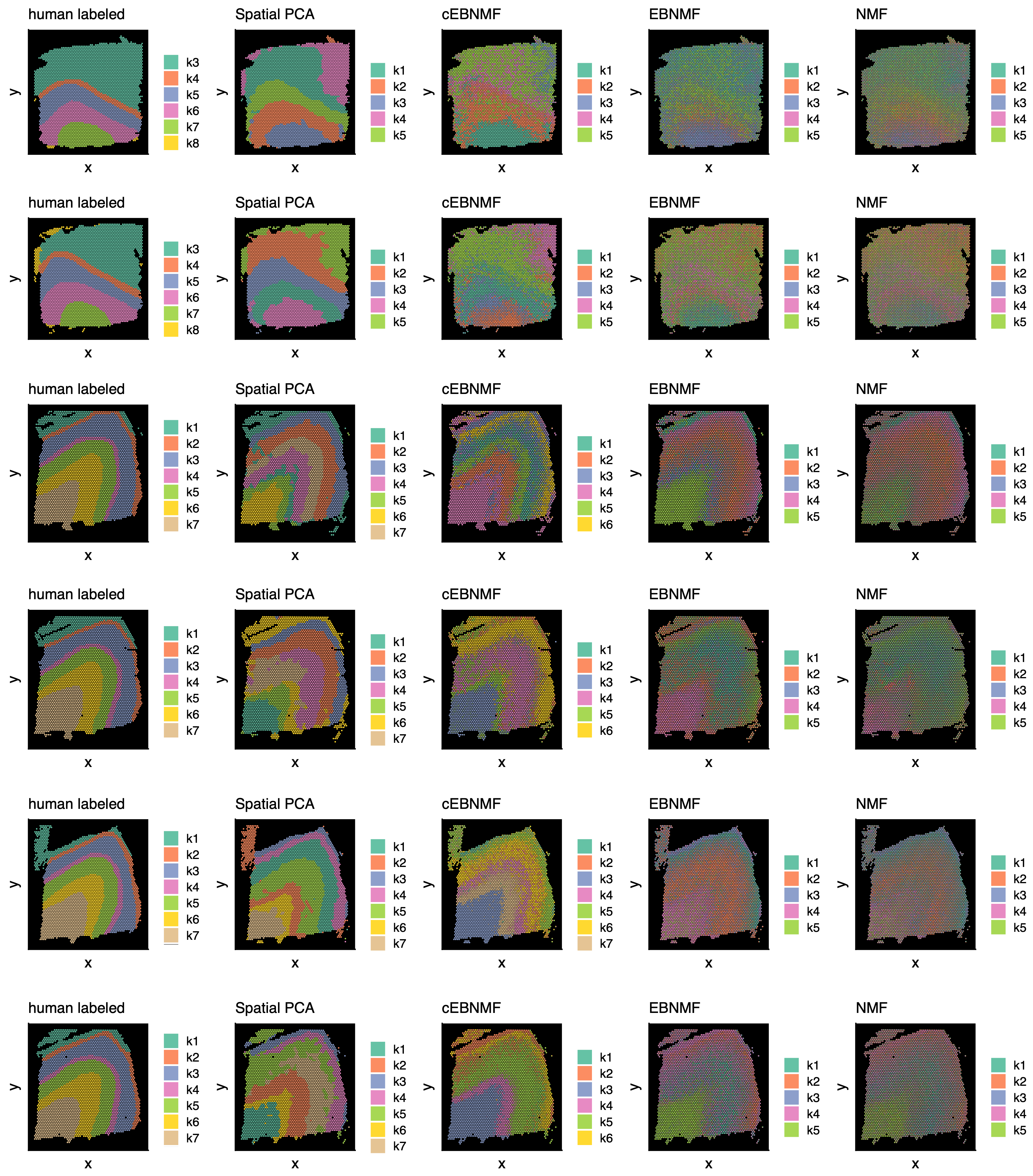}
\caption{Selected results on slices 7 (top row) through 12 (bottom
  row) of the DLPFC spatial transcriptomics data.}
\end{figure}

\end{document}